\documentclass[journal]{IEEEtran}
\usepackage{mathrsfs}
\usepackage{psfrag}
\usepackage{graphicx}
\usepackage{color}
\usepackage{fancyhdr,epsfig}
\usepackage{mathtools}
\usepackage{amsthm}
\usepackage{dsfont,enumerate}

\def\m{\mathcal}

\def\eps{\epsilon}
\def\tn{\textnormal}
\def\wt{\widetilde}
\def\wh{\widehat}
\newcommand{\dfn}{ \stackrel{\tn{def}}{=} }
\newcommand{\markov}{\tn{\mbox{$\multimap\hspace{-0.73ex}-\hspace{-2ex}-$}}}
\newcommand{\Prv}[1]{\Pr\left(#1\right)}

\newcommand{\mw}[1]{#1}

\makeatletter

\let\proof\@undefined
\let\endproof\@undefined
\makeatother

 \usepackage{amsmath, amssymb}

\newtheorem{lemma}{Lemma}
\newtheorem{corollary}{Corollary}
\newtheorem{theorem}{Theorem}
\newtheorem{proposition}{Proposition}
\newtheorem{remark}{Remark}

\newcommand{\set}[1]{\mathcal{#1}}
\allowdisplaybreaks[3]

\begin{document}

\title{On the Capacity of the Discrete Memoryless Broadcast Channel with Feedback}

\author{
\authorblockN{Ofer Shayevitz and Mich\`ele Wigger}\thanks{This paper was in part presented at the \emph{International Symposium on Information Theory 2010}, in Austin, TX, July 2010.

O. Shayevitz was with the Information Theory \& Applications Center, University of California, San Diego, USA \{email: ofersha@ucsd.edu\}. He is now with the D. E. Shaw group, New York, NY. M.~Wigger was with the Electrical and Computer Engineering Department at University of California, San Diego. She is now with the Communications and Electronics Department, at
Telecom ParisTech, Paris, France \{email: michele.wigger@telecom-paristech.fr\}. Her research at the University of California, San Diego, was supported by the Swiss National Science Foundation under Grant~PBEZP2-125703.}}

\maketitle
\thispagestyle{empty}

\begin{abstract}
A coding scheme for the discrete memoryless broadcast channel
with \{noiseless, noisy, generalized\} feedback is proposed, and the associated achievable region derived. The scheme is based on a block-Markov
strategy combining the Marton scheme and a lossy version of the Gray-Wyner scheme with side-information. In each block the transmitter sends fresh data and
update information that allows the receivers to improve the channel
outputs observed in the previous block. For a generalization of Dueck's broadcast channel our scheme achieves the noiseless-feedback capacity, which is strictly larger than the no-feedback capacity. For a generalization of Blackwell's channel and when the feedback is noiseless our new scheme achieves rate points that are  outside the no-feedback capacity region.
It follows by a simple continuity argument that for both these channels and when the feedback noise is sufficiently low, our scheme improves on the no-feedback capacity even when the feedback is noisy.
\end{abstract}

\section{Introduction}\label{sec:intro}
We consider a broadcast channel (BC) with two receivers, where the transmitter has instantaneous access to a feedback signal. Popular examples of such feedback signals are:
\begin{itemize}
\item the channel outputs observed at the two receivers (this setup is called \emph{noiseless feedback}); or
\item a noisy version of these channel outputs (this setup is called \emph{noisy feedback}).
\end{itemize}
Here we allow for very general feedback signals, and only require that the time-$t$ feedback signal is obtained by feeding the time-$t$ input and the corresponding time-$t$ outputs into a memoryless feedback channel. This general form of feedback is commonly referred to as \emph{generalized feedback} \cite{carleial82, willems83, tuninetti07}. For brevity, here we mostly omit the word \emph{generalized}. It is easily seen that our setup includes noiseless feedback and noisy feedback as special cases.

We focus on discrete memoryless broadcast channels (DMBCs), namely where the input and output symbols are from finite alphabets and  the current channel outputs depend on the past inputs and  outputs only through the current input. Our interest lies in the feedback-capacity region of such DMBCs, i.e., in the associated set of rate tuples for which reliable communication is possible.

Most previous results on DMBCs with feedback focus on the case of noiseless
feedback. For example, El Gamal \cite{elgamal78} proved that when the BC is physically
degraded, i.e., one of the two outputs is obtained by processing the other output, then the capacity region with noiseless-feedback coincides with the no-feedback capacity region. In contrast,  Dueck \cite{dueck80} and Kramer \cite{kramer03} described some specific examples of DMBCs where the noiseless-feedback capacity region exceeds the no-feedback capacity region. In Dueck's example, the noiseless-feedback capacity region
is known. However, outside these specific examples, determining the capacity region with feedback for (non-physically-degraded) DMBCs is an open problem. In fact, even characterizing the class of DMBCs where feedback enlarges the capacity region seems hard. This is partly because even the no-feedback capacity region is generally unknown, and partly because a computable single-letter
achievable region for the DMBC with feedback was missing hitherto. Kramer \cite{kramer03} proposed a multi-letter achievable region for the DMBC with noisy or noiseless feedback.

In this paper we propose a coding scheme for the DMBC with
generalized feedback, and present a corresponding single-letter
achievable region.  Subsequently, we analyze two new examples -- a generalization of Dueck's channel \cite{dueck80}, and a noisy version of Blackwell's channel \cite{blackwell} -- where our region is shown to exceed the no-feedback capacity region, even in the presence of feedback noise. Our approach is motivated by Dueck's example
\cite{dueck80}, and is based on the following idea. The transmitter
uses the feedback to identify update information that is useful to the
receivers when decoding their intended messages, and describes
this information in subsequent transmissions. More specifically, our
scheme adopts a block-Markov strategy, where in each block the
transmitter sends a combination of fresh data and compressed update
information pertaining to the data sent in the previous block.
Marton's no-feedback scheme \cite{marton79,GelfandPinsker80} is used in each block to send the fresh data and the update information, at rates outside the no-feedback capacity \mw{region}. The update information sent in a block is essentially an efficient lossy description of the auxiliary inputs in Marton's scheme
from the previous block, taking into account the receivers' observations
and the feedback signal as side-information. The receivers perform backward decoding; starting with the last block, each receiver iteratively performs the following two steps: 1) it decodes its intended data and update information in the current block; and  2) it uses the update information to ``improve'' the channel outputs in the preceding block, which is processed next. This strategy is gainful whenever the cost of the lossy description (i.e., the rate needed to send the update information) is smaller than the increase in rate it supports (i.e., the increase in capacity of the ``improved" channel). Intuitively, this is expected to happen when the descriptions required by the two receivers have a large common part.

Our scheme has some  ideas in common with Lapidoth and Steinberg's scheme for the MAC with strictly causal state-information at the transmitter \cite{lapidothsteinberg10-1, lapidothsteinberg10}.

Recently,  another single-letter achievable region for general DMBCs with feedback has been proposed  \cite{venkataramananpradhan11}\footnote{The conference version of \cite{venkataramananpradhan11} has been presented in the same session at ISIT 2010 as the conference version of this paper, see  \cite{shayevitzwigger10} and \cite{PradhanVenkatar2010}.}. Comparing the achievable region in \cite{venkataramananpradhan11}  to ours however seems difficult.

The paper is organized as follows. In Section \ref{sec:prelim}, the necessary mathematical background is provided. The channel model is described in Section \ref{sec:model}. In Section \ref{sec:Marton}, Marton's scheme for the DMBC without feedback is reviewed in detail. In Section \ref{sec:lossy-gray-wyner}, a lossy version with side-information of the Gray-Wyner distributed source coding setup is introduced, and an achievable region is obtained. The main result of the paper is introduced in Section \ref{sec:main}, where the Marton and the lossy Gray-Wyner schemes are combined into a feedback scheme for general DMBCs, and the associated achievable region is derived. Two new examples are discussed in \ref{sec:example}: A generalization of Dueck's DMBC, and a noisy version of Blackwell's DMBC \cite{blackwell}. In both cases, the region achieved by the new scheme is shown to exceed the no-feedback capacity region, using either noiseless feedback or noisy feedback, in the limit of low feedback noise.

\section{Preliminaries}\label{sec:prelim}
\subsection{Notations}
\mw{We broadly follow the notation in~\cite{El-Gamal--Kim2009}. In particular,} for any real number $M>1$, we use the notation $[M]\dfn \{1,\ldots,\lfloor M\rfloor\}$. The set of positive integers is denoted by $\mathbb{Z}^+$.
\mw{Also, we use upper case symbols to denote random variables, e.g., $A$, and lower case symbols for their realizations, e.g., $a$. The corresponding alphabets are denoted by script symbols, e.g., $\set{A}$; and $|\set{A}|$ is used for the cardinality of $\set{A}$.}
For $n\in\mathbb{Z}^+$ we use $A^n$ and $a^n$ to denote the random sequence $A_1,\ldots, A_n$ and  its realization $a_1,\ldots, a_n$.

We think of a product set of the form $[2^{nr_1}]\times[2^{nr_2}]$ as being one-to-one with $[2^{n(r_1+r_2)}]$, disregarding the associated integer issues throughout. This assumption does not influence our results, as they concern the asymptotic regime $n\to\infty$. For $\eps>0$, we write $\delta(\eps)$ to indicate a general nonnegative function satisfying $\delta(\eps)\to 0$ (arbitrarily slow) as $\eps\to 0$.

A random sequence $X^n$ is said to be \textit{$P_X$-independent-identically distributed ($P_X$-i.i.d.)} if
\begin{equation*}
P_{X^n}(x^n) = \prod_{t=1}^nP_X(x_t)
\end{equation*}
for all $x^n$. Let $(X^n,Y^n)$ be two jointly distributed random sequences, and let $P_{Y|X}$ be some conditional distribution. We say that $Y^n$ is \textit{$P_{Y|X}$-independent given $X^n$} if
\begin{equation*}
P_{Y^n|X^n}(y^n|x^n) = \prod_{t=1}^nP_{Y|X}(y_t|x_t)
\end{equation*}
for all $y^n$ and $x^n$ with $P_{X^n}(x^n)>0$.

We use the notion of typicality as defined in
\cite{El-Gamal--Kim2009}. For a finite alphabet $\m{X}$, a sequence $x^n\in\m{X}^n$ is said to be $\eps$-typical with respect to (w.r.t.) a distribution $P_X$ on $\m{X}$ if
\begin{equation*}
|\pi_{x^n}(x)-P_X(x)| \leq \eps\cdot P_X(x)
\end{equation*}
for all $x\in\m{X}$, where $\pi_{x^n}$ is the distribution over $\m{X}$ corresponding to the relative frequency of symbols in $x^n$. The set of all such sequences is denoted $\m{T}^n_\eps(P_X)$. Similarly, for a law $P_{X_1\cdots X_k}$ over a product alphabet $\m{X}_1\times \cdots \times \m{X}_k$, we denote by $\m{T}^n_\eps(P_{X_1\cdots X_k})$ the set of all $k$-tuples of sequences $(x_1^n\in\m{X}_1^n, \ldots, x_k^n\in\m{X}_k^n)$ that are \emph{jointly} $\eps$-typical w.r.t. $P_{X_1\cdots X_k}$.

Finally, we write $Z\sim\textnormal{Bern}(p)$ for a a binary random variable  taking the values $0$ and $1$ with probabilities $1-p$ and $p$.
\subsection{Basic Lemmas}
The following three lemmas are well known, and used extensively in the sequel.
\begin{lemma}[Conditional Typicality Lemma \cite{El-Gamal--Kim2009}]\label{lem:cond_typ}
Let $P_{XY}$ be some joint distribution. Suppose $x^n\in\m{T}_{\eps'}^n(P_X)$ for some $\epsilon'>0$, and $Y^n$ is $P_{Y|X}$-independent given $X^n=x^n$. Then for every $\eps>\eps'$:
\begin{IEEEeqnarray*}{rCl}
\lim_{n\to\infty}\Pr\big((x^n,Y^n)\not\in\m{T}_\eps^n(P_{XY})\big) = 0.
\end{IEEEeqnarray*}
\end{lemma}

\begin{lemma}[Covering Lemma \cite{El-Gamal--Kim2009}]\label{lem:covering}
Let $0<\eps'<\eps$, and let $X^n$  satisfy $\Pr(X^n\in\m{T}_{\eps'}(P_X))\to 1$  as $n\to\infty$. Also, for each  $n$, let  $M_n\in\mathbb{Z}^+$ be larger than $2^{nr}$ for some $r\geq 0$,  and let $\{Y^n(m)\}_{m=1}^M$ be a set of $P_Y$-i.i.d. sequences such that $\{X^n,\{Y^n(m)\}_{m=1}^M\}$ are mutually independent.  Then, for any law $P_{XY}$ with marginals $P_X$ and $P_Y$ there exists $\delta(\eps)\to 0$ as $\eps\to 0$ such that
\begin{equation*}
\lim_{n\to\infty}\Pr\big(\forall m\in[M]\,,\; (X^n,Y^n(m))\not\in\m{T}_\eps^n(P_{XY})\big) = 0
\end{equation*}
if $r>I(X;Y) +\delta(\eps)$.
\end{lemma}

\begin{lemma}[Packing Lemma \cite{El-Gamal--Kim2009}]\label{lem:packing}
Let $\eps>0$, and $X^n$ be an arbitrary random sequence. Also, for each $n$, let $M_n\in\mathbb{Z}^+$ be smaller than $2^{nr}$ for some $r\geq 0$,  and let $\{Y^n(m)\}_{m=1}^M$ be a set of  $P_Y$-i.i.d. random sequences, where each $Y^n(m)$ is independent of $X^n$.  Then, for any law $P_{XY}$ with marginal $P_Y$ there exists $\delta(\eps)\to 0$ as $\eps\to 0$ such that
\begin{equation*}
\lim_{n\to\infty}\Pr\big(\exists m\in[M]\quad\text{s.t.}\quad (X^n,Y^n(m))\in\m{T}_\eps^n(P_{XY})\big) =0
\end{equation*}
if $r<I(X;Y) - \delta(\eps)$.
\end{lemma}

The following is a simple multivariate generalization of the packing lemma.
\begin{lemma}[Multivariate Packing Lemma]\label{lem:mv_packing}
Let $\eps>0$, and for each  $n$ let  $M_{1,n}, M_{2,n}, M_{3,n}\in\mathbb{Z}^+$ satisfy $M_{i,n}\leq 2^{nr_i}$, for $i\in\{1,2,3\}$. Also, let $\{U_i^n(m)\}_{m=1}^{M_{i,n}}$ be a set of $P_{U_i}$-i.i.d. random vectors such that $\{U_1^n(m_1),U_2^n(m_2),U_3^n(m_3)\}$ are mutually independent for any $m_1,m_2,m_3$. Then, for any law $P_{U_1U_2U_3}$ with marginals $\{P_{U_i}\}_{i=1}^3$, there exists $\delta(\eps)\to 0$ as $\eps\to 0$ such that
\begin{align*}
&\lim_{n\to\infty}\Pr\Big(\exists\;m_i\in[M_i] \;\;\text{for}\;\;i\in\{1,2,3\}\quad{\text s.t.}\\
&\quad\qquad\qquad (U_1^n(m_1),U_2^n(m_2),U_3^n(m_3))\in T_\eps^n(P_{U_1U_2U_3})\Big) \\
& =0
\end{align*}
if
\begin{align} r_1+r_2+r_3 < I(U_1;U_2) + I(U_3;U_1,U_2) - \delta(\eps).\label{eq:cond}
\end{align}
\end{lemma}
\begin{proof}[Proof outline]
Let $\m{E}_{ijk} \dfn \{U_1^n(i),U_2^n(j),U_3^n(k))\in T_\eps^n(P_{U_1U_2U_3})\}$. We need to show that $\Prv{\bigcup_{ijk}\m{E}_{ijk}}\to 0$ under Constraint~\eqref{eq:cond}. By standard typicality/large deviation arguments we have that
\begin{align*}
\Pr(\m{E}_{ijk}) &\leq 2^{-n\left(D(P_{U_1U_2U_3}\|P_{U_1}\times P_{U_3}\times P_{U_3})-\delta(\eps)\right)}\\ & = 2^{-n\left(D(P_{U_1U_2U_3}\|P_{U_1}\times P_{U_3}\times P_{U_3})-\delta(\eps)\right)}
\\
&= 2^{-n\left(D(I(U_1;U_2) + I(U_3;U_1,U_2)-\delta(\eps)\right)}.
\end{align*}
The result follows by taking the union bound over $\m{E}_{ijk}$, and requiring that it tends to zero.
\end{proof}

\section{Channel Model}\label{sec:model}
We consider the discrete memoryless broadcast channel with generalized feedback in Figure~\ref{fig:DMBCFB}.
\begin{figure}
\centering
\psfrag{M}[rc][rc]{\scriptsize $\begin{matrix}M_0\\M_1\\ M_2\end{matrix}$}
\psfrag{M1h}[lc][lc]{\scriptsize $\begin{matrix} \hat{M}_{0,1}\\\hat{M}_1\end{matrix}$}
\psfrag{M2h}[lc][lc]{\scriptsize $\begin{matrix}\hat{M}_{0,2}\\\hat{M}_2\end{matrix}$}
\psfrag{Channel}[cl][cl]{\scriptsize Channel}
\psfrag{Decoder2}[lc][lc]{\scriptsize Receiver2}
\psfrag{Encoder}[lc][lc]{\scriptsize Transmitter}
\psfrag{Decoder1}[lc][lc]{\scriptsize Receiver1}
\psfrag{X}[cc][cc]{\scriptsize$X$}
\psfrag{Y1}[cc][cc]{\scriptsize $Y_{1}$}
\psfrag{Y2}[cc][cc]{\scriptsize  $Y_{2}$}
\psfrag{Yfb}[cc][cc]{\scriptsize $\tilde{Y}$}
\psfrag{W}[lc][lc]{\scriptsize $P_{Y_1Y_2\tilde{Y}|X}$}
 \includegraphics[width=0.9 \columnwidth]{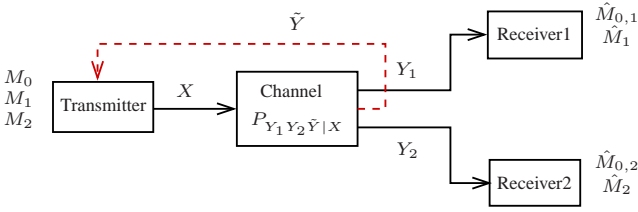}
 \caption{The two-user discrete memoryless BC with generalized feedback.}
 \label{fig:DMBCFB}
 \end{figure}
The goal of the communication is that the transmitter conveys a private Message $M_1$ to a Receiver~1, a private Message $M_2$ to a Receiver~2, and a common message $M_0$ to both receivers. The three messages $M_0, M_1$, and $M_2$ are assumed to be independent and uniformly distributed over the finite sets $[2^{nR_0}]$, $[2^{nR_1}]$, and
$[2^{nR_2}]$ respectively, where $n$ denotes the blocklength and $R_0, R_1, R_2$ are the corresponding common and private transmission rates.

Communication takes place over a DMBC with generalized feedback.
This channel is characterized by a quadruple of finite alphabets
$\m{X}$, $\m{Y}_1, \m{Y}_2,$ and $\wt{\m{Y}}$, and a conditional
probability law $P_{Y_1Y_2\wt{Y}|X}(y_1,y_2,\wt{y}|x)$ where
$x\in\m{X}$, $y_1\in \m{Y}_1$, $y_2 \in \m{Y}_2$, and $\wt{y}\in
\wt{\m{Y}}$. Given that at time $t$ the transmitter feeds the
symbol $x_t$ to the channel, Receiver~1 and Receiver~2 observe the
channel outputs $y_{1,t}\in\m{Y}_1$ and $y_{2,t}\in\m{Y}_2$ respectively, and
the transmitter observes the generalized feedback
$\wt{y}_t\in\wt{\m{Y}}$, with probability
$P_{Y_1Y_2\wt{Y}|X}(y_{1,t},y_{2,t},\wt{y}_t|x_t)$.

Thanks to feedback, the transmitter can produce its time-$t$
channel input $X_t$ as a function of the Messages $M_0,M_1,M_2$ and of
the previously observed feedback outputs
$\wt{Y}^{t-1}\dfn(\wt{Y}_1,\ldots, \wt{Y}_{t-1}):$
\begin{equation}
X_t = \psi^{(n)}_t\left(M_0,M_1, M_2, \wt{Y}^{t-1}\right),
\end{equation}
for some encoding function $\psi^{(n)}_t$, for $t\in\{1,\ldots, n\}$.
The DMBC and its feedback channel are memoryless, which is captured by the following Markov relation for  $t\in[n]$:
\begin{IEEEeqnarray*}{ccccc}
(Y_1^{t-1},Y_2^{t-1},\wt{Y}^{t-1})&\;\markov \;&X_t&\;\markov\;&(Y_{1,t},Y_{2,t},\wt{Y}_t)
\end{IEEEeqnarray*}
where $Y_i^{t-1}\dfn (Y_{i,1},Y_{i,2},\ldots,Y_{i,t-1})$, for $i\in\{1,2\}$.

After $n$ channel uses Receiver~i decodes its intended messages $M_0$ and $M_i$ for $i\in\{1,2\}$. Namely, Receiver~$i$ produces the guess:
\begin{equation}
(\hat{M}_{0,i},\hat{M}_{i} ) = \Psi_{i}^{(n)}(Y_{i}^n),  \quad i\in\{1,2\}
\end{equation}
where $\Psi_i^{(n)}$ denotes Receiver~$i$'s  decoding function.

A rate triplet $(R_0,R_1,R_2)$ is called achievable if for every
blocklength $n$ there exists a set of $n$ encoding functions $\left\{\psi_{t}^{(n)}\right\}_{t=1}^n$ and two decoding functions
$\Psi_1^{(n)}$ and $\Psi_2^{(n)}$ such that the probability of
decoding error, i.e., the probability that
\begin{IEEEeqnarray*}{rCl}
 (M_0,M_1) \neq
  (\hat{M}_{0,1},\hat{M}_1)\; \; \tn{ or } \;\;
(M_0,M_2)\neq
  (\hat{M}_{0,2},\hat{M}_2),
\end{IEEEeqnarray*} tends to 0 as the blocklength $n$
tends to infinity.  The closure of the set of achievable rate triplets
$(R_0,R_1,R_2)$ is called the \emph{feedback capacity-region} of
this setup, and we denote it by $\m{C}_{\tn{GenFB}}$.

The described generalized-feedback setup includes as special cases the
\emph{no-feedback} setup where the feedback outputs are deterministic, e.g.,
$|\m{\wt{Y}}|=1$; the \emph{noiseless-feedback} setup where the feedback
output coincides with the pair of channel outputs, i.e.,
$\wt{Y}=(Y_1,Y_2)$ (see Figure~\ref{fig:noiseless}); and the \emph{noisy-feedback} setup where the
feedback outputs and the channel inputs and outputs satisfy the Markov
relation $X_t \markov (Y_{1,t}, Y_{2,t})\markov \wt{Y}_t$ for all $t\in[n]$ (e.g., the setup in Figure~\ref{fig:noisy}).
In these special cases, we denote the capacity regions by $\m{C}_{\tn{NoFB}}$, $\m{C}_{\tn{NoiselessFB}}$, and $\m{C}_{\tn{NoisyFB}}$, respectively.
\begin{figure}
\centering
\psfrag{M}[rc][rc]{\scriptsize $\begin{matrix}M_0\\M_1\\ M_2\end{matrix}$}
\psfrag{M1h}[lc][lc]{\scriptsize $\begin{matrix} \hat{M}_{0,1}\\\hat{M}_1\end{matrix}$}
\psfrag{M2h}[lc][lc]{\scriptsize $\begin{matrix}\hat{M}_{0,2}\\\hat{M}_2\end{matrix}$}
\psfrag{W}[lc][lc]{\scriptsize $P_{Y_1Y_2\tilde{Y}|X}$}
\psfrag{Channel}[cl][cl]{\scriptsize Channel}
\psfrag{Decoder2}[cc][cc]{\scriptsize Receiver2}
\psfrag{Encoder}[lc][lc]{\scriptsize Transmitter}
\psfrag{Decoder1}[cc][cc]{\scriptsize Receiver1}
\psfrag{X}[cc][cc]{\scriptsize$X$}
\psfrag{Y1}[cc][cc]{\scriptsize $Y_{1}$}
\psfrag{Y2}[cc][cc]{\scriptsize  $Y_{2}$}
 \includegraphics[width=0.9 \columnwidth]{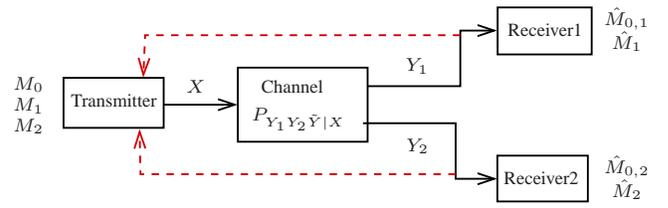}
\caption{The two-user DMBC with noise-free feedback from both outputs.}
\label{fig:noiseless}
\end{figure}
\begin{figure}
\centering
\psfrag{M}[rc][rc]{\scriptsize $\begin{matrix}M_0\\M_1\\ M_2\end{matrix}$}
\psfrag{M1h}[lc][lc]{\scriptsize $\begin{matrix} \hat{M}_{0,1}\\\hat{M}_1\end{matrix}$}
\psfrag{M2h}[lc][lc]{\scriptsize $\begin{matrix}\hat{M}_{0,2}\\\hat{M}_2\end{matrix}$}
\psfrag{W}[lc][lc]{\scriptsize $P_{Y_1Y_2\tilde{Y}|X}$}
\psfrag{Channel}[cl][cl]{\scriptsize Channel}
\psfrag{Decoder2}[lc][lc]{\scriptsize Receiver2}
\psfrag{Encoder}[lc][lc]{\scriptsize Transmitter}
\psfrag{Decoder1}[lc][lc]{\scriptsize Receiver1}
\psfrag{X}[cc][cc]{\scriptsize$X$}
\psfrag{Y1}[cc][cc]{\scriptsize $Y_{1}$}
\psfrag{Y2}[cc][cc]{\scriptsize  $Y_{2}$}
\psfrag{fb-channel}[cc][cc]{\scriptsize Feedback channel}
\psfrag{fb-channel2}[cc][cc]{\scriptsize Feedback channel}
 \includegraphics[width=0.9 \columnwidth]{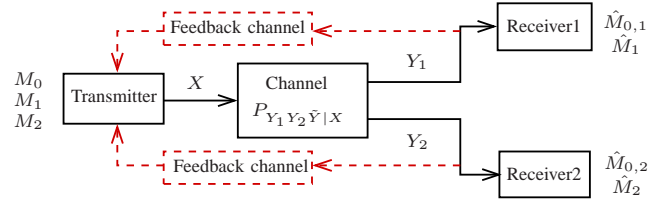}
\caption{Example of a two-user DMBC with noisy feedback.}
\label{fig:noisy}
\end{figure}

\section{Marton's No-Feedback Scheme}\label{sec:Marton}

\mw{We review Marton's achievable region with a common message \cite{marton79,GelfandPinsker80,El-Gamal--Kim2009} and the coding scheme achieving this region.
Redescribing the scheme simplifies the
description of our feedback scheme in Section~\ref{sec:total_scheme}.}

\subsection{Marton's Achievable Region}\label{sec:Marton_Ach}

Let $\m{R}_{\tn{Marton}}$ be the closure of the set of all nonnegative rate triplets $(R_0,R_1,R_2)$ that for some choice of random variables $U_0,U_1,U_2$ over finite alphabets $\m{U}_0$, $\m{U}_1$, $\m{U}_2$  and some function $f\colon \m{U}_0\times \m{U}_1\times \m{U}_2 \to \m{X}$ satisfy
\begin{subequations}\label{eq:marton}
\begin{IEEEeqnarray}{rCl}
R_0+R_1 &<& I(U_0,U_1;Y_1)
\\
R_0+R_2 &<& I(U_0,U_2;Y_2)
\\
R_0+R_1+R_2 &<& I(U_1; Y_1|U_0) + I(U_2; Y_2|U_0) \nonumber \\
&& + \min_{i}I(U_0;Y_i)- I(U_1;U_2|U_0)
\\
2R_0+R_1+R_2 &<& I(U_0,U_1; Y_1) + I(U_0,U_2; Y_2) \nonumber \\
&& - I(U_1;U_2|U_0)\label{eq:marton2r}
\end{IEEEeqnarray}
\end{subequations}
where $X=f(U_0,U_1,U_2)$,
\begin{equation*}
(U_0, U_1,U_2)\markov X \markov (Y_1,Y_2)
\end{equation*}
forms a Markov chain, and  $(Y_1,Y_2)\sim P_{Y_1Y_2|X}$ given $X$.

\begin{theorem}[From \cite{marton79,GelfandPinsker80}]\label{th:marton}
$\m{R}_{\tn{Marton}}\subseteq \m{C}_{\tn{NoFB}}$.
\end{theorem}

\subsection{Marton's Scheme}\label{sec:Marton_scheme}

We describe the scheme for a DMBC $(\mathcal{X}$, $\mathcal{Y}_1$, $\mathcal{Y}_2$, $P_{Y_1Y_2|X})$.
The scheme has parameters $(\mathcal{U}_0$, $\mathcal{U}_1$, $\mathcal{U}_2$, $P_{U_0U_1U_2}$, $f$, $R_0$, $R_{1,p}$, $R_{1,c}$, $R_{2,p}$, $R_{2,c}$, $R_1'$, $R_2'$, $\eps$, $n)$ where
\begin{itemize}
\item $\mathcal{U}_0, \mathcal{U}_1, \mathcal{U}_2$ are auxiliary finite alphabets;
\item $P_{U_0U_1U_2}$ is a joint law over these auxiliary alphabets;
\item $f\colon\mathcal{U}_0\times \mathcal{U}_1\times \mathcal{U}_2\to \mathcal{X}$ is a function mapping the auxiliary inputs into effective inputs;
\item $R_0,R_{1,p}, R_{2,p}, R_{1,c}, R_{2,c}$ are nonnegative communication rates where $R_1\dfn R_{1,p}+R_{1,c}$ and $R_2\dfn R_{2,p}+R_{2,c}$;
\item $R_1',R_2'$ are nonnegative binning rates;
\item $\eps>0$ is a small number; and
\item $n$ denotes the scheme's blocklength.
\end{itemize}

\subsubsection{Code Construction}\label{sec:Marton_code}
 \renewcommand{\ll}{\left \lfloor}
 \renewcommand{\l}{\lfloor}
\renewcommand{\r}{\rfloor}
 \newcommand{\rr}{\right \rfloor}

\mw{Define $R_c\dfn R_0+R_{1,c}+R_{2,c}$.}
The code consists of a single codebook $\mathcal{C}_0$, of $\l2^{nR_c}\r$ codebooks $\{\mathcal{C}_1(\mw{m_c})\}_{\mw{m_c}=1}^{\l2^{nR_c}\r}$, and of $\l2^{nR_c}\r$ codebooks $\{\mathcal{C}_2(\mw{m_c})\}_{\mw{m_c}=1}^{\l2^{nR_c}\r}$.

Codebook $\mathcal{C}_0$ consists of $\l2^{nR_c}\r$ \mw{length-$n$ codewords $\{u_0^n(\mw{m_c})\}_{\mw{m_c}=1}^{\l2^{nR_c}\r}$ whose entries are randomly and independently drawn according
$P_{U_0}$.}
\mw{For $i=1,2$  and $m_c\in[2^{nR_c}]$, Codebook $\mathcal{C}_i(\mw{m_c})$  consists of $\l2^{nR_{i,p}}\r$ bins where each bin $m_{i,p}\in[ 2^{nR_{i,p}}]$ contains  $\l2^{nR_i'}\r$ length-$n$ codewords $\{u_i^n(\mw{m_c},m_{i,p}, \ell_i)\}_{\ell_i=1}^{\l2^{nR_i'}\r}$ that are
randomly drawn $P_{U_i|U_0}$-independent given $u_0^n(\mw{m_c})$. }

\mw{Reveal all codebooks to the transmitter and codebooks $\mathcal{C}_0$ and $\{\mathcal{C}_i(\cdot)\}_{\mw{m_c}=1}^{\l2^{nR_c}\r}$ to Receiver~$i\in\{1,2\}$.}

 \subsubsection{Encoding}\label{sec:Marton_encoding}
 The encoder parses both private messages $M_{1}\in[2^{nR_1}]$ and $M_{2}\in[ 2^{nR_2}]$
 into pairs of independent submessages $(M_{1,p}, M_{1,c})\in [2^{nR_{1,p}}]\times [2^{nR_{1,c}}]$ and $(M_{2,p}, M_{2,c})\in [2^{nR_{2,p}}]\times [2^{nR_{2,c}}]$, and forms the new common message $M_c=(M_0, M_{1,c}, M_{2,c})$ \mw{of rate $R_c$.}

\mw{Now, given that $M_{c}=\mw{m_c}$, \mw{$M_{1,p}=m_{1,p}$, $M_{2,p}=m_{2,p}$}, the encoder makes a list of all pairs $(\ell_1, \ell_2)$ such that}
\footnote{The choice of $\eps/32$ will be helpful later. Here, any $\eps'<\eps$ suffices.}
\begin{IEEEeqnarray}{rCl}\label{eq:enc}
\lefteqn{
(u_0^n(\mw{m_c}), u_1^{n}(\mw{m_c,m_{1,p}},\ell_1) , u_2^{n}(\mw{m_c,m_{2,p}},\ell_2))}\qquad \nonumber \\
&&   \hspace{5cm}\in \mathcal{T}_{\eps/32}^{(n)}(P_{U_0U_1U_2}),\IEEEeqnarraynumspace
\end{IEEEeqnarray}
and chooses one pair from this list at random. We call the chosen pair $(\ell_1^*, \ell_2^*)$. If the list is empty, it chooses $(\ell_1^*,\ell_2^*)$ randomly from the set of all indices $[2^{nR_1'}]\times [2^{nR_2'}]$.

The inputs $x^n$ are obtained from the codewords $u_0^n(\mw{m_c})$, $u_1^n(\mw{m_c, m_{1,p}}, \ell_1^*)$,  $u_2^n(\mw{m_c, m_{2,p}}, \ell_2^*)$ by applying the function $f$ componentwise to these three sequences:
\begin{align*}x_j= f\big(u_{0,j}(\mw{m_c}), u_{1,j}(\mw{m_c,m_{1,p}}, \ell_1^*), u_{2,j}(\mw{m_c,m_{2,p}}, \ell_2^*)\big), \qquad \\
 \hfill\mw{ j\in[n]}.\end{align*}

\subsubsection{Decoding} \label{sec:Marton_decoding}
Given that Receiver~$1$ observes the sequence $y_1^n$, it forms a list of all the tuples $(\mw{\hat{m}_c, \hat{m}_{1,p}}, \hat{\ell}_1)$ that satisfy
\begin{equation}
(u_0^n(\mw{\hat{m}_c}), u_1^{n}(\mw{\hat{m}_c,\hat{m}_{1,p}},\hat{\ell}_1),y_1^n ) \in \mathcal{T}_{\eps}^{(n)}(P_{U_0U_1Y_1}).
\end{equation}
It randomly chooses a tuple $(\mw{\hat{m}_c,\hat{m}_{1,p}}, \hat{\ell}_1)$  from this list
(if the list is empty, it randomly chooses a pair $(\mw{\hat{m}_c,\hat{m}_{1,p}})$ from $[2^{nR_{c}}] \times [ 2^{nR_{1,p}}]$)
and parses $\mw{\hat{m}_c}$ as  $(\hat{m}_{0,1}, \hat{m}_{1,c,1}, \hat{m}_{2,c,1})$. It finally produces  $\hat{m}_{0,1}$ as its guess of message $M_0$ and $\hat{m}_1=(\hat{m}_{1,p}, \hat{m}_{1,c,1})$ as its guess of $M_1$.

 Receiver~$2$ produces its guesses \mw{$\hat{m}_{0,2}$ and $\hat{m}_2$} of the messages $M_0$ and ${M}_2$  in a similar way.

 \subsubsection{Analysis}
\mw{ See Appendix~\ref{sec:Marton_analysis}.}

\section{Lossy Gray-Wyner Coding with Side Information (LGW-SI)}\label{sec:lossy-gray-wyner}
In this section we study a distributed source-coding problem and present an achievable region for this problem. The associated scheme will be used as part of our construction for the DMBC with feedback in Section~\ref{sec:main}.

\begin{figure}
\centering
\psfrag{J0}[rc][rc]{\scriptsize  $K_0$}
\psfrag{J1}[lc][lc]{\scriptsize $K_1$}
\psfrag{J2}[lc][lc]{\scriptsize $K_2$}
\psfrag{Enc}[lc][lc]{\scriptsize  Encoder}
\psfrag{Dec1}[lc][lc]{\scriptsize Decoder1}
\psfrag{Dec2}[lc][lc]{\scriptsize Decoder2}
\psfrag{XY1Y2}[rc][rc]{\scriptsize $X^n$}
\psfrag{Y1}[rc][rc]{\scriptsize $Y_{1}^n$}
\psfrag{Y2}[rc][rc]{\scriptsize  $Y_{2}^n$}
\psfrag{V1}[rc][rc]{\scriptsize $V_{1}^n$}
\psfrag{V2}[rc][rc]{\scriptsize  $V_{2}^n$}
{\includegraphics[width=0.95\columnwidth]{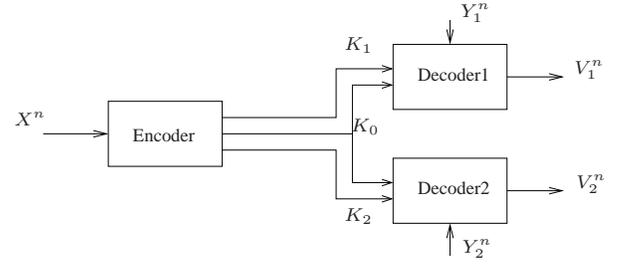}}
\caption{Lossy Gray-Wyner setup with side-information.}
\label{fig:LGW}
\end{figure}

Our source coding problem is depicted in Figure~\ref{fig:LGW}.
Unlike in classical rate-distortion problems where the decoders have to produce sequences that satisfy certain average per-symbol distortion constraints, here, we require that the  sequences produced at the decoders are almost jointly-typical with the source sequence. Thus, our problem is a coordination capacity problem \cite{cuffpermutercover10}.

The rate-distortion problem corresponding to our setup is a lossy version  of the Gray-Wyner distributed source-coding problem in  \cite{graywyner74} with  additional side-information at the decoders. Our achievable region directly leads to an achievable region for this   rate-distortion problem, see \cite{wigger12}.
Special cases of this rate-distortion problem have been considered by    Heegard and Berger  \cite{HeegardBerger85}, Tian and Diggavi  \cite{tiandiggavi08}, and Steinberg and Merhav \cite{steinbergmerhav}, and the lossless counterpart  by
Timo et al. \cite{timograntchankramer08}, \cite{timooechteringwigger12-1}.

\subsection{Setup and Achievable Region}\label{sec:GWSI_ach}

Our setup is parameterized by the tuple $(\mathcal{X}, \mathcal{Y}_1, \mathcal{Y}_2, \mathcal{V}_1, \mathcal{V}_2, P_{XY_1Y_2},P_{V_1|X},P_{V_2|X}, n)$, where
\begin{itemize}
\item $\mathcal{X}, \mathcal{Y}_1, \mathcal{Y}_2, \mathcal{V}_1, \mathcal{V}_2$  are discrete finite alphabets;
\item $P_{XY_1Y_2}$ is a joint probability distribution over the alphabet $\mathcal{X} \times \mathcal{Y}_1\times \mathcal{Y}_2$;
\item $P_{V_1|X}$ and $P_{V_2|X}$ are conditional probability distributions over  $\mathcal{V}_1$ and $\mathcal{V}_2$ given some random variable $X\in\mathcal{X}$;
\item $n$ is the blocklength.
 \end{itemize}

 In the following let $\{(X_t,Y_{1,t},Y_{2,t})\}_{t=1}^n$ be an i.i.d. sequence of triplets of discrete random variables, with marginal distribution $P_{XY_1Y_2}$.  Consider a distributed source coding setting where a sender observes the source sequence $X^n$, Receiver~1 observes the side-information $Y_1^n$, and Receiver~2 observes the side-information $Y_{2}^n$. It is assumed that the sender can noiselessly send three rate-limited messages $K_0, K_1, K_2$ to the receivers:  a common message $K_0$ to both receivers, a private message $K_1$ to Receiver~1 only, and another private message $K_2$ to Receiver~2 only.  More precisely, the encoding procedure is described by an encoding function $\lambda^{(n)} \colon \mathcal{X}^n \rightarrow [2^{nR_0}]\times[2^{nR_1}]\times[2^{nR_2}]$,  which for a sequence $X^n$ produces the messages $(K_0,K_1,K_2) =\lambda^{(n)}(X^n)$.
Each Receiver~$i$, for $i\in\{1,2\}$, produces a reconstruction sequence $\hat{V}_i^n = \Lambda_i^{(n)}(K_0,K_i, Y_i^n)$ by applying a reconstruction function $\Lambda_{i}^{(n)}\colon[2^{nR_0}]\times [2^{nR_i}]\times \mathcal{Y}_i^n\rightarrow \m{V}_i^n$ to the  messages $K_0$ and $K_i$  and the side-information $Y_i^n$.
 The goal of the communication is that for each $i\in\{1,2\}$, the reconstruction sequence $\hat{V}_i^n$ is jointly typical with the source sequence $X^n$ according to $P_X\times P_{V_i|X}$.

  A rate triplet $(R_0,R_1,R_2)$ is said to be \textit{$\eps$-achievable} if there exists a sequence of encoding and reconstruction functions $(\lambda^{(n)}, \Lambda_{1}^{(n)}, \Lambda_{2}^{(n)})$ such that:
\begin{equation*}
\Pr\left((X^{n},\hat{V}_i^{n})\not\in\m{T}^{n}_\eps(P_{XV_i})\right) \to0
\end{equation*}
as $n\to\infty$, for $i\in\{1,2\}$. A triplet is said to be \textit{achievable} if it is $\eps$-achievable for all $\eps>0$. The closure of the set of all achievable rate triplets is denoted $\m{R}_{\tn{LGW}}$.

Let $\m{R}_{\tn{LGW}}^{\tn{inner}}$ be the closure of the set of all nonnegative rate triplets $(R_0,R_1,R_2)$  satisfying
\begin{subequations} \label{eq:GW_FM}
\begin{IEEEeqnarray}{rCl}
R_0 + R_1 &>& I(X;V_0,V_1|Y_1)\label{eq:FM_01}\\
R_0+R_2 &>& I(X;V_0,V_2|Y_2),\label{eq:FM_02}\\
R_0+R_1+R_2 & > & I(X;V_1|Y_1,V_0) +I(X;V_2|Y_2, V_0) \nonumber \\ &&+\max_{i\in\{1,2\}} I(X;V_0|Y_i) \label{eq:FM_012}
\end{IEEEeqnarray}
\end{subequations}
for some choice of  the random variable $V_0$ such that
\begin{equation}\label{eq:GWMarkov}
(V_0,V_1,V_2) \markov X \markov (Y_1,Y_2).
\end{equation}

\begin{theorem}\label{thrm:LGW}
$\m{R}_{\tn{LGW}}^{\tn{inner}}\subseteq \m{R}_{\tn{LGW}}$.
Furthermore,
 $\m{R}_{\tn{LGW}}^{\tn{inner}}$ is convex.
\end{theorem}
\begin{proof}
Inclusion $\m{R}_{\tn{LGW}}^{\tn{inner}}\subseteq \m{R}_{\tn{LGW}}$ is established  in Section~\ref{sec:GWSI_scheme}.
The convexity of
$\m{R}_{\tn{LGW}}^{\tn{inner}}$ is proved in Appendix~\ref{sec:app_conv}. \end{proof}
\mw{Notice that the region depends on the joint conditional distribution $P_{V_1V_2|V_0X}$ only through the marginal conditional distributions $P_{V_1|V_0X}$ and $P_{V_2|V_0X}$.
}

\subsection{Scheme}\label{sec:GWSI_scheme}

In this section we describe a scheme achieving the region $\set{R}_{\textnormal{LGW}}^{\textnormal{inner}}$.
Our scheme is  similar to Heegard and Berger's scheme for the Wyner-Ziv setup with several, differentely informed receivers \cite[Theorem~2]{HeegardBerger85}. However, our scheme also uses the double-binning technique for the common codebook proposed in \cite{tiandiggavi08}, but where here the double-binning is performed in two different ways, one way that is relevant for Receiver~1 and the other way relevant for Receiver~2. This is beneficial when the quality of the side-information at the two receivers is very different.

The  scheme we propose has parameters $\mathcal{V}_0$, $P_{V_0 V_1 V_2|X}$,  $R_{0,0}$, $R_{0,1}$, $R_{0,2}$, $R_{1,0}$, $R_{1,1}$, $R_{2,0}$, $R_{2,2}$, ${R}_0'$, $R_1'$, $R_2'$, $\eps$, $n$, where
\begin{itemize}
\item $\mathcal{V}_0$ is an auxiliary alphabet;
\item $P_{V_0 V_1 V_2|X }$ is a conditional joint probability distribution over $\mathcal{V}_0\times \mathcal{V}_1\times \mathcal{V}_2$ given some $X\in\mathcal{X}$ such that its marginals satisfy $\sum_{v_0,v_2} P_{V_0V_1V_2|X}(v_0,v_1,v_2|x)= P_{V_1|X}(v_1|x)$ and $\sum_{v_0,v_1} P_{V_0V_1V_2|X}(v_0,v_1,v_2|x)=P_{V_2|X}(v_2|x)$;
\item $R_{0,0}, R_{0,1}, R_{0,2}, R_{1,0}, R_{1,1}, R_{2,0}, R_{2,2}\geq 0$ are nonnegative communication rates;
\item ${R}_0', R_1', R_2'\geq0$ are nonnegative binning rates, where ${R}_0'$ cannot be smaller than $\max\{R_{1,0}, R_{2,0}\}$;
\item $\eps>0$ is a small number; and
\item $n$ is the scheme's blocklength.
\end{itemize}

\subsubsection{Codebook Generation}\label{sec:GWSI_code}
Generate three codebooks $\m{C}_0, \m{C}_1, \m{C}_2$ independentely of each other in the following way.

Codebook $\m{C}_0$ consists of $\l2^{nR_{0,0}}\r$ superbins,
each containing $\l2^{n{R}_{0}'}\r$  length-$n$ codewords whose entries are
randomly and independently generated according to
the law $P_{V_0}$.
\begin{figure}
 \centering
 \psfrag{Bin1}[cc][cc]{\scriptsize Bin 1}
  \psfrag{Bin2}[cc][cc]{\scriptsize Bin 2}
  \psfrag{Bin_1}[cb][ct]{\scriptsize a subbin of user~{\color{red}1}}
  \psfrag{Bin_2}[cc][cc]{\scriptsize a subbin of user~{\color{blue}2} }
 \psfrag{Superbin1}[cc][cc]{\scriptsize superbin 1}
  \psfrag{Superbin2}[cc][cc]{\scriptsize superbin 2}
  \psfrag{Codebook V0}[cc][cc]{}
  \psfrag{Codebook V1}[cb][ct]{\scriptsize Codebook $V_1^n$}
   \psfrag{Codebook V2}[cb][ct]{\scriptsize Codebook $V_2^n$}
   \psfrag{dots}[cb][ct]{\scriptsize $\vdots$}
{\includegraphics[width=0.25\textwidth]{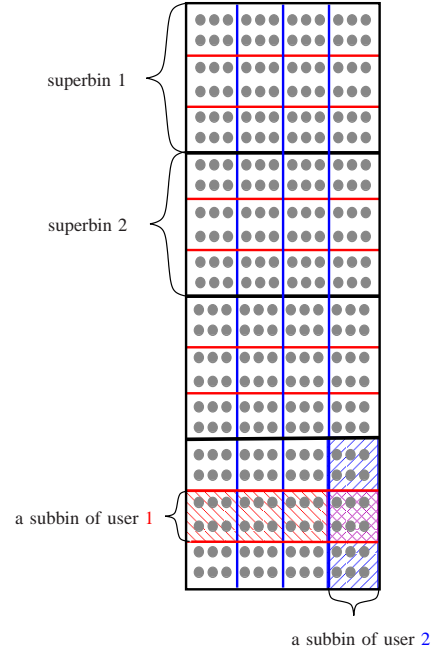}}
\caption{Double-binning structure of codebook $\set{C}_0$ in our lossy Gray-Wyner scheme with side-information. The dots depict the codewords.}
\label{superbin}
\end{figure}
\mw{We make two partitions of the codewords in each superbin, see Figure~\ref{superbin}. In the first partition
the codewords of each superbin are assigned to
$\l2^{n{R}_{1,0}}\r$ subbins, each containing $\l2^{n({R}_0'-R_{1,0})}\r$ codewords; in the second partition they are assigned to
$\l2^{n{R}_{2,0}}\r$ subbins, each containing $\l2^{n({R}_0'-R_{2,0})}\r$ codewords.} There are thus two different ways
to refer to a specific codeword in $\m{C}_0$.
When we consider the first partition, we denote the
codewords in the $k_{1,0}\in[2^{n{R}_{1,0}}]$-th
subbin of  superbin $k_{0,0}\in[2^{nR_{0,0}}]$ by
\[\{v_{0}^n(1; k_{0,0}, k_{1,0},\ell_{1,0})\}_{\ell_{1,0}=1}^{\l2^{n({R}_0'-R_{1,0})}\r};\]
when we consider the second
partition, we denote the
codewords in the   $k_{2,0}\in[2^{n{R}_{2,0}}]$-th subbin of  superbin
$k_{0,0}\in[2^{nR_{0,0}}]$ by
\[\{v_{0}^n(2;k_{0,0}, k_{2,0},\ell_{2,0})\}_{\ell_{2,0}=1}^{\l2^{n({R}_0'-R_{2,0})}\r}.\]
Thus, here the first index indicates whether
the last two indices refer to the first or the second partition of the superbins.

\mw{For $i\in\{1,2\}$, Codebook $\m{C}_i$} consists of  $\l2^{nR_{0,i}}\r$ superbins each containing $\l2^{nR_{i,i}}\r$ subbins with $\l2^{nR_{i}'}\r$ codewords of length $n$, where all entries of all codewords are randomly and independently drawn according to $P_{V_i}$.
For $k_{i,i}\in[2^{nR_{i,i}}]$, we denote the codewords in the $k_{i,i}$-th subbin of superbin $k_{0,i} \in[2^{nR_{0,i}}]$ by \[\{v_{i}^n(k_{0,i}, k_{i,i}, \ell_i)\}_{\ell_i=1}^{\l2^{nR_i'}\r}.\]

All codebooks are revealed to the sender, and codebooks $\{\m{C}_0,\m{C}_i\}$  are revealed to Receiver~$i\in\{1,2\}$.

\subsubsection{LGW-SI Encoder}\label{sec:GWSI_encoder}
Given that the encoder observes the source sequence $X^n=x^n$, it searches the codebooks  $\m{C}_0,\m{C}_1,\m{C}_2$ for a triplet of codewords $v_0^{n}(1;k_{0,0}, k_{1,0}, \ell_{1,0})\in\m{C}_0$, $v_1^n(k_{0,1}, k_{1,1},\ell_1)\in\m{C}_1$, $v_2^n(k_{0,2}, k_{2,2},\ell_2)\in\m{C}_2$ such that for $i\in\{1,2\}$:
\begin{equation}\label{eq:typcond}
(X^n,v_0^n(1;k_{0,0}, k_{1,0}, \ell_{1,0}),v_i^n(k_{0,i}, k_{i,i},\ell_{i}))\in\m{T}_{\eps/2}^n(P_{XV_0V_i}).
\end{equation}
It then forms a list of all tuples of indices $(k_{0,0}, k_{1,0},  \ell_{1,0}, k_{0,1}, k_{1,1}, \ell_1, k_{0,2}, k_{2,2},\ell_2)$ satisfying \eqref{eq:typcond}. If the
list is non-empty, the sender  chooses one tuple from this list at random; otherwise,
it randomly chooses a tuple
$(k_{0,0}, k_{1,0},  \ell_{1,0}, k_{0,1}, k_{1,1}, \ell_1, k_{0,2}, k_{2,2}, \ell_2)$ from the set
$[2^{nR_{0,0}}]\times [2^{nR_{1,0}}]\times [2^{n({R}_0'-R_{1,0})}]\times[2^{nR_{0,1}}]\times [2^{nR_{1,1}}]\times [2^{nR_1'}]\times [2^{nR_{0,2}}]\times [2^{nR_{2,2}}]\times [2^{nR_2'}]$.
We denote the chosen indices by $k_{0,0}^*, k_{1,0}^*, \ell_{1,0}^*, k_{0,1}^*, k_{1,1}^*, \ell_1^*, k_{0,2}^*, k_{2,2}^*, \ell_2^*$. Also, define $(k_{2,0}^*, \ell_{2,0}^*)$ such that  $v_{0}^n(2;k_{0,0}^*, k_{2,0}^*,\ell_{2,0}^*)$ and $v_{0}^n(1;k_{0,0}^*, k_{1,0}^*,\ell_{1,0}^*)$ refer to the same codeword in $\mathcal{C}_0$.

The encoder then sends the product message $K_0= (k_{0,0}^*, k_{0,1}^*,  k_{0,2}^*)$  to both receivers, the product message $K_1=(k_{1,0}^*, k_{1,1}^*)$  to Receiver~1 only, and the product message $K_2=(k_{2,0}^*,k_{2,2}^*)$ to Receiver~2 only.

\subsubsection{LGW-SI Decoder}\label{sec:GWSI_decoder}
Receiver $i\in\{1,2\}$ first parses the common message $K_0$ as $(K_{0,0}, K_{0,1}, K_{0,2})$ and its private message $K_i$ as $K_i=(K_{i,0}, K_{i,i})$. Then, given that Receiver~$i$'s
side-information is $Y_i^n=y_i^n$ and that  $K_{0,0}=k_{0,0}$,  $K_{0,i}=k_{0,i}$,  $K_{i,0}=k_{i,0}$, and
$K_{i,i}=k_{i,i}$, Receiver~$i$ seeks a codeword ${v}_0^{n}(i;k_{0,0}, k_{i,0}, \ell_{i,0})$ in codebook $\m{C}_0$ and a codeword
${v}_i^{n}(k_{0,i}, k_{i,i}, \ell_i)$ in codebook $\m{C}_{i}$ such that
\[({v}_0^{n}(i; k_{0,0}, k_{i,0},\ell_{i,0}),{v}_i^{n}(k_{0,i}, k_{i,i},  \ell_i),y_i^{n})\in\m{T}_{\eps}^n(P_{V_0V_iY_i}).\] If exactly one such pair of codewords
exists, Receiver~$i$ produces as its reconstruction sequence  $\hat{V}_{i}^n= {v}_i^{n}(k_{0,i}, k_{i,i}, \ell_i)$. Otherwise, it randomly chooses  a
 triplet $( k_{0,i}', k_{i,i}',  \ell_i')$ from the set  $[2^{nR_{0,i}}]\times [2^{nR_{i,i}}]\times [2^{nR_i'}]$ and  produces  as its reconstruction sequence
 $\hat{V}_{i}^n= {v}_i^{n}(k_{0,i}', k_{i,i}', \ell_i')$.

\mw{\subsubsection{Analysis}
In  Appendix~\ref{sec:GWSI_analysis} we show that under Constraints~\eqref{eq:GW_FM} the failure probability of our scheme tends to 0 as $n\to \infty$.
 The existence of a deterministic coding scheme with vanishing failure probability follows from standard arguments. }

\section{Main Result \mw{for DMBCs with Generalized Feedback}}\label{sec:main}

\subsection{Achievable Region}\label{sec:achregion}
Consider a DMBC with generalized feedback given by $\m{X}, \m{Y}_1, \m{Y}_2, \tilde{\m{Y}}, P_{Y_1Y_2\tilde{Y}|X}$. Let $\m{R}_{\tn{inner}}$ be the closed convex hull of the set of all nonnegative triplets $(R_0,R_1,R_2)$ that satisfy Inequalities (\ref{eq:inner}) shown on top of the next page, for some choice of auxiliary random variables $(U_0,U_1,U_2, V_0, V_1, V_2)$ and function $f$ such that $X={f}(U_0,U_1,U_2)$,
\begin{align}\label{eq:M31}
(V_0,V_1,V_2)\markov (U_0,&U_1,U_2,\wt{Y}) \markov (Y_1,Y_2)
\end{align}
and
\begin{equation}\label{eq:M41}
(U_0,U_1,U_2)\markov X\markov (Y_1,Y_2,\wt{Y})
\end{equation}
form Markov chains, and $(Y_1,Y_2, \tilde{Y})\sim P_{Y_1Y_2\tilde{Y}|X}$.

Notice that for noise-free feedback where $\tilde{Y}=(Y_1,Y_2)$ the Markov chain \eqref{eq:M31} is satisfied for any choice of the auxiliary random variables $(U_0,U_1,U_2, V_0,V_1,V_2)$.
\begin{figure*}
\begin{subequations}\label{eq:inner}
\begin{IEEEeqnarray}{rCl}
 R_0 +R_1 &\leq& I(U_0,U_1;Y_1,V_1) - I(U_0,U_1,U_2,\wt{Y};V_0,V_1|Y_1)\label{eq:inner1}
\\
 R_0+R_2 &\leq& I(U_0,U_2;Y_2,V_2) - I(U_0,U_1,U_2,\wt{Y};V_0,V_2|Y_2)\label{eq:inner2}
\\
 R_0+R_1+R_2 &\leq& I(U_1; Y_1,V_1|U_0) + I(U_2; Y_2,V_2|U_0) + \min_{i\in\{1,2\}}I(U_0;Y_i,V_i)- I(U_1;U_2|U_0)\nonumber \\
 && -I(U_0,U_1,U_2,\wt{Y};V_1|V_0,Y_1) -I(U_0,U_1,U_2,\wt{Y};V_2|V_0,Y_2) - \max_{i\in\{1,2\}} I(U_0,U_1,U_2,\wt{Y};V_0|Y_i)
\\
\nonumber 2R_0+R_1+R_2 &\leq& I(U_1U_0; Y_1,V_1) + I(U_2,U_0; Y_2,V_2) - I(U_1;U_2|U_0) \\
&& - I(U_0,U_1,U_2,\wt{Y};V_0,V_1|Y_1) -I(U_0,U_1,U_2,\wt{Y};V_0,V_2|Y_2)\label{eq:inner4}
\end{IEEEeqnarray}
\end{subequations}
\rule{\textwidth}{0.6pt}
\end{figure*}

\begin{theorem}\label{thrm:innerbound2}
$\m{R}_{\tn{inner}} \subseteq \m{C}_{\tn{GenFB}}$.
\end{theorem}
The proof of the theorem is given in Subsection~\ref{sec:total_scheme}. A few remarks are in order:
\begin{remark}
The region $\m{R}_{\tn{inner}}$
includes $\m{R}_{\tn{Marton}}$, because when for a given choice of $(U_0,U_1,U_2)$, constraints \eqref{eq:inner}
are specialized to $(V_0,V_1,V_2) = {\rm const}$, then it results in  the Marton region~\eqref{eq:marton}. The inclusion is also clear from the construction of our scheme in Subsection~\ref{sec:total_scheme} ahead.
\end{remark}

\begin{remark}
In our coding scheme we can allow $f$ to be a \emph{randomized} function. In this case, the scheme achieves the region $\m{R}_{\tn{inner}}$ but where the input $X$ can be an arbitrary random variable satisfying the Markov chain \eqref{eq:M41}.

\mw{We can also superposition all the codebooks on a $P_Q$-i.i.d. random vector $Q^n$ that is known at the transmitter and both receivers. In this case, the joint typicality checks need to be modified accordingly. The new scheme achieves a region as in $\m{R}_{\tn{inner}}$ but where the mutual information constraints~\eqref{eq:inner} need to be conditioned on $Q$ and the Markov chains in \eqref{eq:M31} and \eqref{eq:M41} require $Q$ in the middle position.}

 It is not clear whether \mw{these changes result} in an improved region compared to $\m{R}_{\tn{inner}}$.
\end{remark}
\begin{remark}
Recall that for fixed finite alphabets, the Shannon information measures are continuous (say w.r.t. Euclidean distance) in the joint distribution \cite{Yeung_book}. Fix the channel's input, output, and feedback alphabets. Then for any fixed choice of $(P_{U_0U_1U_2}, f, P_{V_0V_1V_2|U_0U_1U_2\wt{Y}})$, the quantities on the right-hand side of  Inequalities (\ref{eq:inner}) are continuous in $P_{Y_1Y_2\wt{Y}|X}$.
\end{remark}

\begin{remark}\label{rem:contin}
By the previous remark, the following conclusion holds for any DMBC $P_{Y_1Y_2|X}$ with feedback alphabet $\m{\wt{Y}} = \m{Y}_1\times\m{Y}_2$. Assume that the region $\m{R}_{\tn{inner}}$  associated with noiseless feedback (i.e., $\tilde{Y}=(Y_1,Y_2)$) strictly contains $\m{C}_{\tn{NoFB}}$. Now, if we consider a \emph{noisy} feedback channel $P_{\wt{Y}|XY_1Y_2}$ that is close enough to  the noiseless feedback (i.e.,  $\tilde{Y}$ close  to $(Y_1,Y_2)$), then also the region  $\m{R}_{\tn{inner}}$ associated with this noisy feedback strictly contains  $\m{C}_{\tn{NoFB}}$.
\end{remark}

\subsection{Scheme achieving $\m{R}_{\tn{innner}}$}\label{sec:total_scheme}

\begin{figure*}
\psfrag{ldots}{$\ldots$}
\psfrag{newdata}{\small fresh data}
\psfrag{updateinfo}{\small update info.}
\psfrag{t=1}{\footnotesize $t\!=\!1$}
\psfrag{t=np}{\footnotesize $t\!=\!n$}
\psfrag{t=nB}[r][c]{\footnotesize $t\!=Bn+\gamma n$}
\psfrag{M1}{\scriptsize $\!\!M_{0,(1)}\!,\!M_{1,(1)}\!,\!M_{2,(1)}$}
\psfrag{M2}{\scriptsize $\!\!M_{0,(2)}\!,\!M_{1,(2)}\!,\!M_{2,(2)}$}
\psfrag{M3}{\scriptsize $\!\!M_{0,(3)}\!,\!M_{1,(3)}\!,\!M_{2,(3)}$}
\psfrag{M5}{}
\psfrag{M6}{\scriptsize $\!\!M_{0,(B)}\!,M_{1,(B)}\!,M_{2,(B)}$}
\psfrag{V1}{\scriptsize $\!\!K_{0,(1)}\!,K_{1,(1)}\!,K_{2,(1)}$}
\psfrag{V2}{\scriptsize $\!\!K_{0,(2)}\!,K_{1,(2)}\!,K_{2,(2)}$}
\psfrag{V5}{}
\psfrag{V7}{\scriptsize $\!\!K_{0,(\!B)}\!,\!K_{1,(\!B)}\!,\!K_{2,(\!B)}$}
\psfrag{V6}{\scriptsize $\!K_{0,(\!B\!-\!1)}\!,\!K_{1,(\!B\!-\!1)}\!,\!K_{2,(\!B\!-\!1)}$}
\psfrag{Block1}{\small Block 1}
\psfrag{Block2}{\small Block 2}
\psfrag{Block3}{\small Block 3}
\psfrag{BlockBm}{\small Block $B-1$}
\psfrag{BlockB}{\small Block $B$}
\psfrag{BlockB1}{\small Block $B+1$}
 \includegraphics[width=\textwidth]{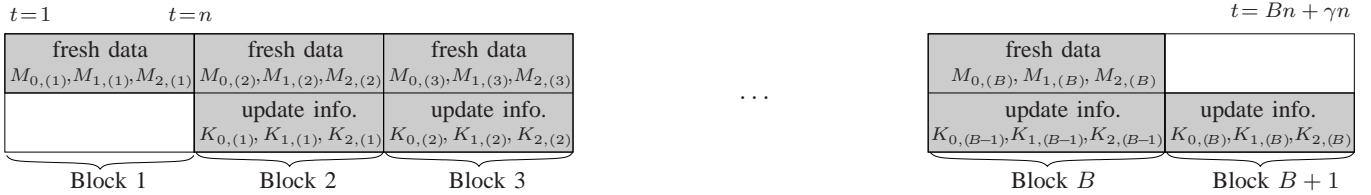}
    \caption{Block-Markov strategy of our feedback-scheme.}
    \label{fig2}

\end{figure*}

 Our scheme combines Marton's no-feedback scheme of Section~\ref{sec:Marton_scheme} with our LGW-SI scheme of Section~\ref{sec:GWSI_scheme} using a block-Markov framework. We first present the high-level idea of the scheme, which is also depicted in Figure~\ref{fig2}.
Transmission takes place over $B+1$ consecutive blocks, where the first $B$ blocks are of length $n$ each, and the last block is of length $\gamma n$ for $\gamma>1$. We denote the input/output/feedback sequences in Block $b\in[B]$ by $X^n_{(b)},Y^n_{i,(b)},\wt{Y}^n_{(b)}$, respectively, and the input/output sequences in Block $B+1$ by  $X^{n'}_{(B+1)},Y^{n'}_{i,(B+1)}$. The messages to be sent are in a product form $M_i  = (M_{i,(1)},\ldots,M_{i,(B)})$, for $i\in\{0,1,2\}$, where each $M_{i,(b)}$ is uniformly distributed over the set $[2^{nR_i}]$.
The effective rates of transmission are thus
\begin{equation}\label{eq:effrates}
\left(\frac{B}{B+ \gamma} R_0, \frac{B}{B+ \gamma} R_1, \frac{B}{B+ \gamma} R_2\right)
\end{equation}
and approach $(R_0, R_1, R_2)$ as the number of blocks $B\to\infty$.

In each block $b$ the transmitter uses Marton's no-feedback  scheme to send the Messages $M_{0,(b)}, M_{1,(b)}, M_{2,(b)}$ together with update information $K_{0,(b-1)}, K_{1,(b-1)}, K_{2,(b-1)}$ pertaining to the messages sent in the previous block. An exception is the first (resp. last) block where only the message tuple (resp. update information) is sent. The update information is constructed in a way that when $(K_{0,(b)}, K_{i,(b)})$ is available at Receiver~$i$, the latter can use it to ``improve'' its block-$b$ observations $Y_{i,(b)}^n$. This facilitates the decoding of the corresponding messages $M_{0,(b)}, M_{1,(b)}, M_{2,(b)}$, which otherwise might not have been possible to decode reliably. The update information is generated via the LGW-code described in Section~\ref{sec:GWSI_scheme}. The code is designed for an LGW-setup where the encoder's \mw{``source sequence"} consists of the auxiliary Marton-codewords and the feedback signal, and where the receivers' \mw{``side-informations"} consist of their respective channel outputs.

Each Receiver $i$, for $i\in\{1,2\}$, performs backward decoding. It starts from the last block and decodes the update information $(K_{0,(B)}, K_{i, (B)})$ based on
$Y_{i,(B+1)}^{n'}$. Denote its guess by $\hat{K}_{0,i, (B)}, \hat{K}_{i,(B)}$. Then, for each block $b\in[B]$, starting from block $B$ and going backwards, it performs the following steps:
\begin{enumerate}
\item Using  $(\hat{K}_{0,b}, \hat{K}_{i,b})$, it ``improves'' its block-$b$ outputs $Y_{i,(b)}^n$.
\item Based on these ``improved'' outputs, it then decodes the data  $(M_{0,(b)}, M_{i,(b)})$ and the update information $(K_{0,(b-1)}, K_{i,(b-1)})$. We denote the corresponding guesses by  $(\hat{M}_{0,(b)}, \hat{M}_{i,(b)})$ and $(\hat{K}_{0,i,(b-1)}, \hat{K}_{i,(b-1)})$.
\end{enumerate}

We now describe the coding scheme in more detail. Our scheme has  parameters $(\mathcal{U}_0$, $\mathcal{U}_1$, $\mathcal{U}_2$, $\mathcal{V}_0$, $\mathcal{V}_1$, $\mathcal{V}_2$,
$P_{U_0U_1U_2}$, $f$, $P_{V_0V_1V_2| U_0U_1 U_2 \wt{Y}}$,  $ {R}_0$, ${R}_{1}$, ${R}_2$, $\bar{R}_1'$, $\bar{R}_2'$, $\tilde{R}_0$, $\tilde{R}_1$, $\tilde{R}_2$, $\tilde{R}_0'$, $\tilde{R}_1'$, $\tilde{R}_2'$, $\eps$, $\gamma$, $n$, $B$), where:
\begin{itemize}
\item $\mathcal{U}_0$, $\mathcal{U}_1$, $\mathcal{U}_2$, $\mathcal{V}_0$, $\mathcal{V}_1$, and $\mathcal{V}_2$ are finite auxiliary alphabets;
\item$P_{U_0U_1U_2}$ is a joint probability law over $\mathcal{U}_0\times \mathcal{U}_1 \times \mathcal{U}_2$;
\item $f$ is a function $f: \mathcal{U}_0\times \mathcal{U}_1\times \mathcal{U}_2\to \mathcal{X}$;
\item $P_{V_0V_1V_2| U_0U_1 U_2 \wt{Y}}$ is a conditional probability law over $ \mathcal{V}_0\times \mathcal{V}_1 \times \mathcal{V}_2$ given a tuple $(U_0,U_1,U_2, \wt{Y})$;
\item ${R}_0,{R}_1,{R}_2, \tilde{R}_0 , \tilde{R}_1, \tilde{R}_2$ are nonnegative communication rates;
\item $\bar{R}_1', \bar{R}_2',\tilde{R}_0', \tilde{R}_1', \tilde{R}_2'$ are nonnegative binning rates;
\item $\eps>0$ is a small number; and
\item $n$, $\gamma$,  and $B$ are  positive integers determining the scheme's blocklength.
\end{itemize}

\subsubsection{\mw{Code Construction}}
For each block $b\in[B]$ we construct a Marton code
 for a DMBC with parameters $(\mathcal{X}, \mathcal{Y}_1 \times \mathcal{V}_1, \mathcal{Y}_2\times \mathcal{V}_2, P_{(Y_1V_1)(Y_2 V_2)|X})$ using the code construction in Subsection~\ref{sec:Marton_code}. As parameters of  this construction we choose:
\begin{itemize}
\item the auxiliary alphabets $\mathcal{U}_0, \mathcal{U}_1, \mathcal{U}_2$;
\item the joint law  $P_{U_0,U_1,U_2}$ over these alphabets;
\item the function $f\colon \mathcal{U}_0\times \mathcal{U}_1\times \mathcal{U}_2\to \mathcal{X}$;
\item the nonnegative
communication rates $\bar{R}_0$, $\bar{R}_{1,p}$, $\bar{R}_{2,p}$, $\bar{R}_{1,c}$, $\bar{R}_{2,c}$ where we require that $\bar{R}_0={R}_0+\tilde{R}_0$, $\bar{R}_{1,p}+\bar{R}_{1,c}={R}_1+\tilde{R}_1$, and $\bar{R}_{2,p} +\bar{R}_{2,c}=R_2 +\tilde{R}_2$;
\item the nonnegative binning rates $\bar{R}_1', \bar{R}_2'$;
\item the small number $\eps$; and
\item the blocklength $n$.
\end{itemize}

For block $B+1$, we use a Marton   scheme  for the DMBC with parameters $(\mathcal{X}, \mathcal{Y}_1,\mathcal{Y}_2, P_{Y_1Y_2|X})$ of block length $\gamma n$ where the scheme is chosen as to achieve the rate triplet $(\gamma^{-1}\wt{R}_0, \gamma^{-1}\wt{R}_1,\gamma^{-1}\wt{R}_2)$.  To make sure that such a scheme exists, we assume throughout the proof that the single-user channels $P_{Y_1|X}$ and $P_{Y_2|X}$  both have positive capacities.\footnote{When one of the two single-user channels has capacity 0, then the broadcast problem is not very interesting. In fact, in this case both the capacity regions with noiseless feedback and with no-feedback are degenerate.
}
Under this assumption, it is readily verified that for $\gamma>1$ large enough such a scheme exists.

In what follows, let $\varphi_{(b)}, \Phi_{1,(b)}, \Phi_{2,(b)}$ denote the encoding and decoding rules corresponding to the Marton-code in block $b$, for any $b\in[B+1]$. Also, let the triplet $(U_{0,(b)}^n, U_{1,(b)}^n, U_{2,(b)}^n)$ denote the auxiliary codewords produced by the block-$b$ Marton encoder $\varphi_{(b)}$, for any $b\in[B]$, and let $X_{(b)}^n, Y_{1,(b)}^n, Y_{2,(b)}^n$ and $\wt{Y}_{(b)}^n$ denote the corresponding blocks of \mw{channel inputs/outputs/feedback outputs.}

Then, consider the
LGW-SI setup with the following parameters:
\begin{itemize}
\item the source alphabet $(\mathcal{U}_0\times  \mathcal{U}_1\times \mathcal{U}_2\times  \tilde{\mathcal{Y}})$;
\item the decoder side-information alphabets $\mathcal{Y}_1$ and $\mathcal{Y}_2$;
\item the reconstruction alphabets $\mathcal{V}_1$ and $\mathcal{V}_2$;
\item the source-side-information law $P_{(U_0U_1U_2 \tilde{Y})Y_1 Y_2}$; and
\item the reconstruction laws $P_{V_1|U_0U_1U_2\tilde{Y}}(v_1|u_0,u_1,u_2, \tilde{y})= \sum_{ v_0, v_2}P_{V_0V_1V_2|U_0U_1U_2 \tilde{Y} }(v_0, v_1, v_2| u_0,u_1,u_2, \tilde{y})$ and  $P_{V_2|U_0U_1U_2\tilde{Y}}(v_2|u_0,u_1,u_2, \tilde{y})= \sum_{ v_0, v_1}P_{V_0V_1V_2|U_0U_1U_2 \tilde{Y} }(v_0, v_1, v_2| u_0,u_1,u_2, \tilde{y})$.
\end{itemize}
For this LGW-SI setup we construct for each block $b\in[B]$ an LGW-SI code  as described in Subsection \ref{sec:GWSI_code}. Our construction has the following parameters:
\begin{itemize}
\item the auxiliary alphabet $\mathcal{V}_0$;
\item the conditional law $P_{V_0V_1V_2|U_0U_1U_2 \tilde{Y} }$;
\item the nonnegative rates $\tilde{R}_{0,0}$, $\tilde{R}_{0,1}$, $\tilde{R}_{0,2}$, $\tilde{R}_{1,0}$, $\tilde{R}_{1,1}$, $\tilde{R}_{2,0}$, $\tilde{R}_{2,2}$, $\tilde{R}_0'$, $\tilde{R}_1'$, $\tilde{R}_2'$;
\item the binning rates  $\tilde{R}_0', \tilde{R}_1', \tilde{R}_2'\geq0$, where $\tilde{R}_0'$ cannot be smaller than $\max\{\tilde{R}_{1,0}, \tilde{R}_{2,0}\}$;
\item the small number $\eps/2$; and
\item the blocklength $n$.
\end{itemize}
In what follows, let $\lambda_{(b)}$, $\Lambda_{1,(b)}$,  and $\Lambda_{2,(b)}$ denote the LGW-SI  encoding and decoding  rules corresponding to these codes.

\subsubsection{Encoding} In the first block $b=1$, the transmitter forms the product messages  $J_{0,(1)}\dfn(M_{0,(1)},1)$, $J_{1,(1)}\dfn(M_{1,(1)},1)$, and $J_{2,(1)}\dfn(M_{2,(1)}, 1)$, and applies the Marton encoding rule $\varphi_{(1)}$ to this triplet $J_{0,(1)}, J_{1,(1)}, J_{2,(1)}$.

In blocks $b\in{2,\ldots,B}$ the transmitter first applies the LGW-SI encoding function $\lambda_{(b-1)}$ to its \mw{``source sequence"} $(U_{0,(b-1)}^n, U_{1,(b-1)}^n, U_{2,(b-1)}^n, \tilde{Y}^n_{(b-1)})$ to generate the update messages $(K_{0, (b-1)}, K_{1,(b-1)}, K_{2,(b-1)})$.  (Recall that $U_{0,(b-1)}^n, U_{1,(b-1)}^n, U_{2,(b-1)}^n$  denote the Marton auxiliary codewords produced in the previous  encoding step.) The transmitter then generates the messages $J_{i,(b)}\dfn (M_{i,(b)},K_{i,(b-1)})$, and encodes them via the Marton encoding rule $\varphi_{(b)}$. It finally sends the outcome of this encoding over the channel.

In the last block $B+1$, the transmitter first applies the LGW-SI encoding function $\lambda_{(B)}$ to the sequences $(U_{0,(B)}^n, U_{1,(B)}^n, U_{2,(B)}^n, \tilde{Y}^n_{(B)})$ to generate the update messages $K_{i,(B)}$, for $i\in\{0,1,2\}$. It then forms the tuple $J_{0,(B+1)}\dfn(1,K_{0,(B)})$, $J_{1,(B+1)}\dfn(1,K_{1,(B)})$, and $J_{2,(B+1)}\dfn(1,K_{2,(B)})$ and encodes them via the Marton encoding rule $\varphi_{(B+1)}$.

\subsubsection{Decoding at Receiver~$i$}
Decoding is performed backwards, starting from the last block. Receiver~$i$ first applies the decoding rule $\Phi_{i,(B+1)}$ to the outputs $Y^{n}_{i,(B+1)}$ attempting to decode the indices
 $(J_{0,(B+1)},J_{i,(B+1)})$, and parses its guess  $(\hat{J}_{0,i,(B+1)},\hat{J}_{i,(B+1)})$ as $\hat{J}_{0,i,(B+1)} = (1,\hat{K}_{0,i,(B)})$ and $\hat{J}_{i,(B+1)} = (1,\hat{K}_{i,(B)})$.

Now, for every block $b\in\{\mw{1}, \ldots, B\}$, starting with block $B$ and going backwards, the receiver performs the following steps. It applies the LGW-SI decoder
$\Lambda_{i, (b)}$ to its guess of the update messages  $(\hat{K}_{0,i, (b)}, \hat{K}_{i,(b)})$ obtained in block $b+1$, and to its \mw{``side-information"} $Y_{i,(b)}^n$. It then applies Marton's decoding rule $\Phi_{i,(b)}$ to the pair  $(Y_{i,(b)}^n, \hat{V}_{i,(b)}^n)$, where $\hat{V}_{i,(b)}^n$ denotes the reconstruction sequence produced by the LGW-SI decoder $\Lambda_{i,(b)}$. Finally, it
parses the   guess  produced by Marton's decoding rule $(\hat{J}_{0,i,(b)},\hat{J}_{i,(b)})$ as $\hat{J}_{0,i,(b)} = (\hat{M}_{0,i,(b)},\hat{K}_{0,i,(\mw{b-1})})$ and $\hat{J}_{i,(b)} = (\hat{M}_{i,(b)},\hat{K}_{i,(\mw{b-1})})$.

Receiver~$i$'s guess of the messages $M_0$ and $M_i$ are the products $\hat{M}_{0,i} = (\hat{M}_{0,i, (1)}, \ldots, \hat{M}_{0,i, (B)})$ and $\hat{M}_i = (\hat{M}_{i, (1)}, \ldots, \hat{M}_{i, (B)})$.

\mw{\subsubsection{Analysis}
In  Appendix~\ref{eq:comerror} we show that under Constraints~\eqref{eq:inner} the error probability  of our scheme tends to 0 as $n\to \infty$.
 The existence of a deterministic coding scheme with vanishing error probability follows from standard arguments. }

\section{Examples}\label{sec:example}
\subsection{The Generalized Dueck DMBC} \label{sec:example_dueck}

In \cite{dueck80} Dueck presented the first example of a DMBC where noise-free feedback increases  capacity. In his setup, the channel input consists of three bits, $X=(X_0,X_1,X_2)$, and each of the two outputs of two bits, $Y_1=(Y_{1,1}, Y_{1,0})$ and
$Y_2=(Y_{2,0},Y_{2,2})$ where
\begin{IEEEeqnarray*}{rCl}
Y_{1,0}=Y_{2,0} &=& X_0, \\
Y_{1,1}&=& X_1 \oplus Z,\\ Y_{2,2} &=&  X_2 \oplus Z.
\end{IEEEeqnarray*}
Here, the noise $Z$ is Bern(1/2) and independent of the inputs, and $\oplus$ denotes
addition modulo 2.

Obviously, without feedback, the outputs $Y_{1,1}$ and $Y_{2,2}$ are useless. Thus, the no-feedback-capacity is given by the set of all nonnegative rate triplets $(R_0,R_1,R_2)$ satisfying
\begin{equation*}
R_0+ R_1+R_2\leq 1.
 \end{equation*}
With noiseless feedback, the capacity is increased.
\begin{theorem}[Dueck \cite{dueck80}]
The noiseless feedback capacity of Dueck's DMBC is given by the set of all nonnegative rate triplets $(R_0,R_1,R_2)$ satisfying
\begin{IEEEeqnarray}{rCl}
R_0+R_1 & \leq & 1 \qquad \textnormal{and} \qquad R_0+R_2  \leq  1.
\end{IEEEeqnarray}
\end{theorem}
\begin{proof} The converse follows from the cutset bound. The achievability by the following simple blocklength-$(n+1)$ scheme. The transmitter sends   lossless descriptions of the Message pairs $(M_0, M_1)$ and $(M_0,M_2)$ using the inputs $\{X_{1,t}\}_{t=1}^n$ and $\{X_{2,t}\}_{t=1}^n$, respectively. Additionally,  for $t=2,\ldots, (n+1)$, it repeats the previous noise symbol as $X_{0,t}=Z_{t-1}$. The transmitter knows $Z_{t-1}$ at time $t$  because it is cognizant of the input $X_{1,t-1}$ (or $X_{2,t-1}$) and, through the feedback, also of $Y_{1,t-1}=X_{1,t-1}+Z_{t-1}$ (or $Y_{2,t-1}=X_{2,t-1}+Z_{t-1}$).

Notice that each Receiver~$i\in\{1,2\}$ learns the noise sequence $\{Z_t\}_{t=1}^n$ from its sequence of outputs $\{Y_{i,0,t}\}_{t=2}^{n+1}$. Receiver~$i$ can thus compute the channel inputs $X_{i,t}= Y_{i,i,t}-Z_{t}$, for $t=1, \ldots, n$, and recover the desired pair of messages $(M_0,M_i)$ whenever the sum-rate $R_0+R_i$ is smaller than $\frac{n}{n+1}$. Letting the block-length $n$ tend to infinity, we get the desired achievability result.
\end{proof}

We generalize Dueck's setup to the DMBC depicted in Figure~\ref{fig:Gendueck}. We assume that all three
binary channels are noisy, and the first and third
channels are corrupted by different noises.
\begin{figure}
\centering
             \vspace{0.2cm}
\psfrag{M}[rc][rc]{\scriptsize $\begin{matrix} M_0\\M_1\\ M_2\end{matrix}$}
\psfrag{Mh1}[lc][lc]{\scriptsize $\begin{matrix} \hat{M}_{0,1}\\\hat{M}_1\end{matrix}$}
\psfrag{Mh2}[lc][lc]{\scriptsize $\begin{matrix} \hat{M}_{0,2}\\\hat{M}_2\end{matrix}$}
\psfrag{Channel}[cl][cl]{\scriptsize Channel}
\psfrag{Receiver2}[lc][lc]{\scriptsize Receiver 2}
\psfrag{Encoder}[lc][lc]{\scriptsize Transmitter}
\psfrag{Receiver1}[lc][lc]{\scriptsize Receiver 1}
\psfrag{X}[cc][cc]{\scriptsize$X$}
\psfrag{Y1}[cc][cc]{\scriptsize $Y_{1}$}
\psfrag{Y2}[cc][cc]{\scriptsize  $Y_{2}$}
\psfrag{W}[lc][lc]{\scriptsize $W(\cdot,\cdot|\cdot)$}
\psfrag{B1}[cc][cc]{\scriptsize $X_{1}$}
\psfrag{B0}[cc][cc]{\scriptsize $X_{0}$}
\psfrag{B2}[cc][cc]{\scriptsize $X_{2}$}
\psfrag{Z_0}[cc][cc]{\scriptsize$Z_{0}$}
\psfrag{Z_1}[cc][cc]{\scriptsize$Z_{1}$}
\psfrag{900}[cc][cc]{\scriptsize $Y_{01}$}
\psfrag{909}[cc][cc]{\scriptsize $Y_{11}$}
\psfrag{0}[cc][cc]{\scriptsize $Y_{02}$}
\psfrag{990}[cc][cc]{\scriptsize $Y_{22}$}
\psfrag{Z_2}[cc][cc]{\scriptsize$Z_{2}$}
{\includegraphics[width=0.9 \columnwidth]{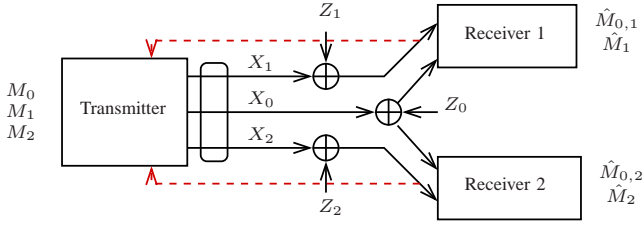}}
\caption{Generalization of Dueck's DMBC with feedback example.}
\label{fig:Gendueck}
\end{figure}
 Thus, as before, the channel input consists of three bits,
$X=(X_1,X_0,X_2)$, and each output of two bits,
$Y_1=(Y_{1,1}, Y_{1,0})$ and $Y_{2}=(Y_{2,0},Y_{2,2})$. However, now,
\begin{IEEEeqnarray*}{rCl}
Y_{1,0}=Y_{2,0}  &=  &X_0  \oplus Z_0,\\  Y_{1,1}&  = & X_1 \oplus Z_1, \\ Y_{2,2} &=&
  X_2  \oplus Z_2,
\end{IEEEeqnarray*}
where $Z_0,Z_1,Z_2$ are binary random variables of a given joint
law $P_{Z_0Z_1Z_2}$.

\begin{proposition}\label{obs:nofb}
The no-feedback capacity region of the
  generalized Dueck DMBC is the set of all nonnegative rate triplets
  $(R_o,R_1,R_2)$ that satisfy
  \begin{subequations}\label{eq:nofb}
\begin{IEEEeqnarray}{rCl}
R_0+R_1 & \leq & 2- H(Z_0,Z_1),\\
R_0+R_2 & \leq & 2- H(Z_0, Z_2),\\
R_0+R_1+R_2 & \leq & \mw{3- H(Z_0,Z_1)-H(Z_0, Z_2)}. \IEEEeqnarraynumspace
\end{IEEEeqnarray}
\end{subequations}
\end{proposition}
\begin{proof}
The no-feedback capacity of a DMBC depends on the channel law $P_{Y_1Y_2|X}(y_1,y_2|x)$ only through the marginal laws $P_{Y_1|X}(y_1|x)$ and $P_{Y_2|X}(y_2|x)$ (see e.g., \cite{sato78}). We therefore assume in the following that $Z_2\markov Z_0 \markov Z_1$. The converse follows then simply by applying the cutset bound to this modified setup. The achievability follows from Marton's achievable region. More precisely, if in the region in \eqref{eq:marton} we choose $U_0, U_1, U_2$ to be i.i.d. Bern$(1/2)$ and $X_i=U_i$, for $i\in\{0,1,2\}$, then it evaluates to our region in \eqref{eq:nofb}. (Notice that since we choose $U_0, U_1, U_2$ independent,  constraint~\eqref{eq:marton2r} on $2R_0+R_1+R_2$ is not active.)
\end{proof}

\mw{Our scheme in Section~\ref{sec:total_scheme} allows us to obtain the capacity region for the Generalized Dueck DMBC with noiseless feedback when}
\begin{IEEEeqnarray}{rCl}
H(Z_{0},Z_{1}) &\leq &1 \qquad  \tn{and} \qquad
H(Z_{0}, Z_{2}) \leq 1.\label{eq:condition}
\end{IEEEeqnarray}

\begin{theorem}\label{th:1}
Under condition~\eqref{eq:condition} and when no common message is sent, i.e.,  $R_0=0$, the noiseless-feedback capacity of the Generalized Dueck DMBC
  is the set of all nonnegative rate pairs $(R_1,R_2)$
  satisfying
  \begin{subequations}\label{eq:cap1}
\begin{IEEEeqnarray}{rCl}
R_1 & \leq & 2- H(Z_0,Z_1),  \\
R_2 & \leq & 2- H(Z_0,Z_2),  \\
R_1+R_2 &\leq & 3 - H(Z_0,Z_1,Z_2).
\end{IEEEeqnarray}
\end{subequations}
\end{theorem}
\begin{proof}
  The converse follows from the cutset bound.  The direct part
  follows from Theorem~\ref{thrm:innerbound2} by taking the convex
  hull of the achievable regions that result when \eqref{eq:inner} is evaluated for the following two choices:
  $(U_0,U_1,U_2)$ i.i.d. Bern(1/2); $X_i=U_i$ for $i\in\{0,1,2\}$; $V_i=(X_0,X_i)$ for $i\in\{1,2\}$; and either
  $V_0=(Z_0,Z_1)$ or $V_0=(Z_0,Z_2)$. (Notice that since $U_0, U_1, U_2$ are independent, Constraint \eqref{eq:inner4} is subsumed by Constraints \eqref{eq:inner1} and \eqref{eq:inner2}.)
\end{proof}

In view of Proposition~\ref{obs:nofb}, we have the following corollary to Theorem~\ref{th:1}.
\begin{corollary}\label{cor:gener-dueck-dmbc}
If the triplet $(Z_0,Z_1, Z_2)$ satisfies \eqref{eq:condition}  and  does not form the Markov chain $Z_{1}- Z_{0} - Z_{2}$, then noiseless feedback strictly increases
the capacity of our Generalized Dueck DMBC.
\end{corollary}

Let's briefly consider the case of noisy feedback $\wt{Y}=(Y_{1,1} \oplus W_1, Y_{1,0}\oplus W_0, Y_{2,2} \oplus W_2)$ where $(W_0,W_1, W_2)$ are arbitrary distributed binary random variables, with marginals $W_i\sim {\rm Bern}(q_i)$, for $q_0, q_1, q_2\in(0,1)$. Evaluating Theorem~\ref{thrm:innerbound2} for this noisy-feedback setup is cumbersome and left out. But from Corollary~\ref{cor:gener-dueck-dmbc} and the continuity considerations mentioned in Remark~\ref{rem:contin}, we can conclude the following.
\begin{remark}
If the noise triplet $(Z_0, Z_1, Z_2)$ satisfies \eqref{eq:condition}  and does not form the Markov chain $Z_{1}-Z_{0}-Z_{2}$,   then for any sufficiently small value of $\max\{q_0,q_1,q_2\}$, the noisy feedback introduced above enlarges the capacity region of the Generalized Dueck DMBC.
\end{remark}

\subsection{The Noisy Blackwell DMBC}\label{sec:example_blackwell}
Consider the noisy version of the Blackwell DMBC \cite{blackwell} in Figure~\ref{fig:NoisyBlackwell}.
\begin{figure}
\centering
\psfrag{M}[rc][rc]{\scriptsize $\begin{matrix}M_0\\M_1\\ M_2\end{matrix}$}
\psfrag{M1}[lc][lc]{\scriptsize $\begin{matrix} \hat{M}_{0,1}\\\hat{M}_1\end{matrix}$}
\psfrag{M2}[lc][lc]{\scriptsize $\begin{matrix}\hat{M}_{0,2}\\\hat{M}_2\end{matrix}$}
\psfrag{Channel}[cl][cl]{\scriptsize Channel}
\psfrag{Receiver2}[lc][lc]{\scriptsize Receiver2}
\psfrag{Transmitter}[lc][lc]{\scriptsize Transmitter}
\psfrag{Receiver1}[lc][lc]{\scriptsize Receiver1}
\psfrag{X}[cc][cc]{\scriptsize$X$}
\psfrag{Y1}[cc][cc]{\scriptsize $Y_{1}$}
\psfrag{Y2}[cc][cc]{\scriptsize  $Y_{2}$}
\psfrag{1}[cc][cc]{\tiny $1$}
\psfrag{0}[cc][cc]{\tiny $0$}
\psfrag{2}[cc][cc]{\tiny $2$}
\psfrag{Z}[cc][cc]{\scriptsize$Z$}
\psfrag{W}[cc][cc]{\scriptsize$\tilde{Z}$}
\psfrag{Blackwell Channel}[cc][cc]{\scriptsize Blackwell Channel}
{\includegraphics[width=0.95 \columnwidth]{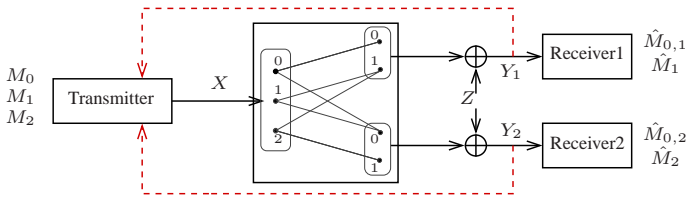}}
\caption{A noisy version of Blackwell's DMBC with noiseless feedback.}
\label{fig:NoisyBlackwell}
\end{figure}
The input alphabet is ternary $\m{X} = \{0,1,2\}$ and both output alphabets are binary $\m{Y}_1=\m{Y}_2 = \{0,1\}$. Let $Z\sim\textnormal{Bern}(p)$, with $p<\frac{1}{2}$,  be independent of $X$. The channel law $P_{Y_1Y_2|X}$ is described as follows.
\begin{align}\label{eq:outputs}
Y_1 = \left\{\begin{array}{ll}Z & X=0 \\ 1-Z & X=1,2\end{array}\right.\quad Y_2 = \left\{\begin{array}{ll}Z & X=0,1 \\ 1-Z & X=2.\end{array}\right.
\end{align}
When $p=0$,  the described DMBC specializes to Blackwell's DMBC. For this case the capacity region with and without feedback is given by Marton's region.
We consider noiseless feedback and present an achievable region for this setup based on the region $\set{R}_{\textnormal{Inner}}$  in Theorem~\ref{thrm:innerbound2}.

Let $U_0, U_1, U_2$ be binary random variables, where $U_0\sim\tn{Bern}(\frac{1}{2})$, and where given $U_0=0$  the pair $(U_1,U_2)$ has  joint conditional law
\begin{table}[h!]
\centering
\normalsize
$P_{U_1U_2|U_0=0}$: \qquad \begin{tabular}{c||c|c}
& $U_2=0$& $U_2=1$ \\
\hline\hline
$U_1=0$ & $\phantom{\tilde{\beta}}\alpha\phantom{\tilde{\beta}}$ & 0 \\
$U_1=1$& $1-\alpha-\beta$& $\beta$
\end{tabular}
\end{table}
\newline
 for some nonnegative $\alpha,\beta$ satisfying $\alpha+\beta\leq 1$,  and given $U_0=1$ it has joint conditional law
 \begin{table}[h!]
 \centering
 \normalsize
 $P_{U_1U_2|U_0=1}$: \qquad
 \begin{tabular}{c||c|c}
& $U_2=0$& $U_2=1$ \\
\hline
\hline
$U_1=0$ & $\phantom{\tilde{\beta}}\beta\phantom{\tilde{\beta}}$ & 0 \\
$U_1=1$& $1-\alpha-\beta$& $\alpha$
\end{tabular}
\end{table}
\newline Set $X\dfn U_1+U_2$ (real addition), and let $V_1\dfn U_1$,  $V_2\dfn U_2$, and
$V_0 \dfn V_1\oplus {Y}_1 = Z$.
Evaluating the region in \eqref{eq:inner} for this choice of random variables, we obtain the following theorem.
\begin{theorem}\label{th:blackwell}
All nonnegative rate triplets $(R_0,R_1,R_2)$  satisfying
\begin{align*}
&R_0+R_1 \leq h_b\left(\left(\frac{\alpha+\beta}{2}\right)\star p\right) - h_b(p)
\\
&R_0+R_2 \leq h_b\left(\left(\frac{\alpha+\beta}{2}\right)\star p\right) - h_b(p)
\\
&R_0+R_1+R_2 \leq h_b\left(\left(\frac{\alpha+\beta}{2}\right)\star p\right) + \frac{1-\beta}{2}h_b\left(\frac{\alpha}{1-\beta}\right)  \\ &\hspace{50pt}+ \frac{1-\alpha}{2}h_b\left(\frac{\beta}{1-\alpha}\right)- h_b(p)\\
&2R_0+R_1+R_2 \leq 2h_b\left(\left(\frac{\alpha+\beta}{2}\right)\star p\right) -  2  h_b(p)\nonumber \\ & \hspace{50pt}+ H\left([\alpha, \beta, 1-\alpha-\beta]\right) - h_b(\alpha)-h_b(\beta)
\end{align*}
are achievable over the Noisy Blackwell DMBC. Here,  $H\big([p_1, \ldots, p_m]\big)\dfn\sum_{i=1}^m p_i\log \frac{1}{p_i}$; $h_b(p)\dfn H([p, 1-p])$;  and $\gamma \star p\dfn (1-\gamma) p + \gamma(1-p)$.
\end{theorem}
Let us consider the sum-rates $R_1+R_2$ guaranteed by the region above. To that end, we set $R_0=0$ and note it is sufficient to consider only the last two inequalities.
We get the following corollary to Theorem~\ref{th:blackwell}.
\begin{corollary}
With noiseless feedback, our scheme achieves all nonnegative rate pairs $(R_1,R_2)$ satisfying Inequality~\eqref{eq:C_bw_f} shown on top of the next page.
\begin{figure*}
\begin{IEEEeqnarray}{rCl}\label{eq:C_bw_f}
R_1+R_2\geq \sup_{\substack{\alpha, \beta\geq 0 \colon \\ \alpha+\beta \leq 1}} \;\; \min\Bigg\{& &h_b\left(\left(\frac{\alpha+\beta}{2}\right)\star p\right)+ \frac{1-\beta}{2}h_b\left(\frac{\alpha}{1-\beta}\right) + \frac{1-\alpha}{2}h_b\left(\frac{\beta}{1-\alpha}\right)- h_b(p),\nonumber \\ &&\quad 2h_b\left(\left(\frac{\alpha+\beta}{2}\right)\star p\right) + H\left([\alpha, \beta, 1-\alpha-\beta]\right) - h_b(\alpha)-h_b(\beta)-  2  h_b(p) \Bigg\}
\end{IEEEeqnarray}
\hrule
\end{figure*}
\end{corollary}
For comparison, let us now upper bound the sum-rates $R_1+R_2$ that are achievable without feedback. Since the no-feedback capacity of a DMBC depends only on the marginals $P_{Y_1|X},P_{Y_2|X}$ \cite{sato78}, the capacity region for the Noisy Blackwell channel remains the same if in the definitions of $Y_1$ and $Y_2$ (see \eqref{eq:outputs}) we replace $Z$ by  independent $\textnormal{Bern}(p)$ random variables $Z_1$ and $Z_2$, respectively. Computing the cut-set upper bound for this latter setting, we obtain that all rate pairs $(R_1,R_2)$ that are achievable without feedback must satisfy
\begin{IEEEeqnarray}{rCl}
\lefteqn{R_1+R_2}\nonumber \\ &\leq &\!\!\sup_{\alpha\in(0,\frac{1}{2})} \left\{H\big([\alpha(p-\bar{p})^2+p\bar{p},\bar{p}^2+2\alpha \bar{p}(p-\bar{p}),\right.\nonumber\\ &&\left. \hspace{1.7cm}p^2+2\alpha {p}(\bar{p}-p), \alpha(p-\bar{p})^2+p\bar{p}]\big)\right\} - 2h_b(p),\nonumber \\\label{eq:bb2}
\end{IEEEeqnarray}
where $\bar{p}\dfn 1-p$.  Figure~\ref{fig1} depicts the bounds
\eqref{eq:C_bw_f} and \eqref{eq:bb2} together with
a cut-set upper bound on the sum-rates $R_1+R_2$ that are achievable with noiseless feedback. By this Figure~\ref{fig1}:
\begin{corollary}
Noiseless feedback enlarges the capacity  region of the Noisy Blackwell-DMBC.
\end{corollary}
\begin{remark}
Let $\wt{Y}=(Y_1\oplus W_1,Y_2\oplus W_2)$, where $(W_1,W_2)$ are jointly distributed binary random variables with marginals $W_i\sim\tn{Bern}(q_i)$, mutually independent of $(X,Y_1,Y_2)$. By the continuity argument in Remark~\ref{rem:contin}, for any $p\in(0,1)$ and small enough $\max\{q_1,q_2\}$, noisy feedback strictly enlarges the capacity region of the Noisy Blackwell-DMBC with noisy feedback.
\end{remark}

\begin{figure}
\vspace{-3cm}

\psfrag{R}{\footnotesize $R_{1}+R_2$}
\psfrag{p}{\footnotesize $p$}
\psfrag{0}[r][r]{\tiny 0}
\psfrag{0a}[t][t]{\tiny 0}
\psfrag{0.15}[t][t]{}
\psfrag{0.25}[t][t]{}
\psfrag{0.35}[t][t]{}
\psfrag{0.45}[t][t]{}
\psfrag{0.05}[t][t]{}
\psfrag{0.1}[t][t]{\tiny 0.1}
\psfrag{0.2a}[t][t]{\tiny 0.2}
\psfrag{0.2}[r][r]{\tiny 0.2}
\psfrag{0.3}[t][t]{\tiny 0.3}
\psfrag{0.4}[t][t]{\tiny 0.4}
\psfrag{0.5}[t][t]{\tiny 0.5}
\psfrag{0.6}[r][r]{\tiny 0.6}
\psfrag{0.8}{}
\psfrag{1}[r][r]{\tiny 1}
\psfrag{1.4}[r][r]{\tiny 1.4}
\psfrag{1.6}{}
\psfrag{1.2}{}
\psfrag{0.46}{}
\psfrag{0}{\tiny 0}
\psfrag{0}{\tiny 0}
    \epsfig{file=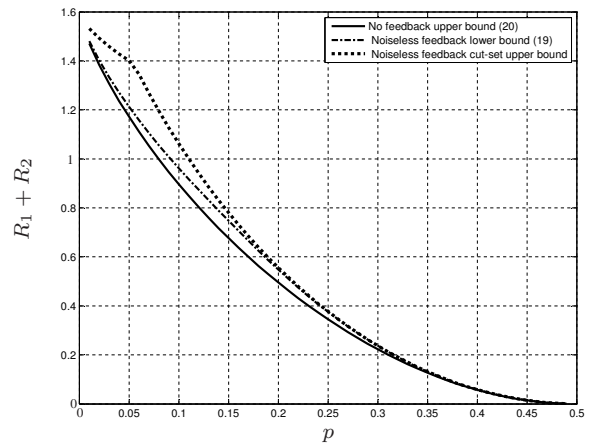, scale = 0.42}
\vspace{-3.3cm}

    \caption{Bounds on the maximum sum-rates $R_1+R_2$ that are achievable over the Noisy Blackwell DMBC with no feedback and noiseless feedback.}
    \label{fig1}

\end{figure}
\section*{Acknowledgement}
The authors thank the \mw{Associate Editor, the anonymous reviewers, and Prof.~Daniela~Tuninetti for the careful reading of the manuscript and their valuable comments.} In particular for pointing them to \cite{tiandiggavi08}.

\appendices

\section{Analysis of Marton's Scheme}
\label{sec:Marton_analysis}

We analyze the average probability of error of \mw{Marton's} scheme averaged over the random messages, codebooks, and channel realizations, see also \cite{marton79,GelfandPinsker80,El-Gamal--Kim2009}. Recall that an error occurs whenever
\begin{equation*}
{(\hat{M}_{0,1}, \hat{M}_1)\neq (M_0,M_1) \textnormal{ or } (\hat{M}_{0,2}, \hat{M}_2) \neq (M_0,M_2)}.
\end{equation*}
By the symmetry of the code construction
\begin{IEEEeqnarray*}{rCl}
\Pr{[ \textnormal{error}]}=\Pr{[\textnormal{error}|M_c=M_{1,p}=M_{2,p}=1]}.
\end{IEEEeqnarray*}
To shorten notation we denote the event that $M_c=M_{1,p}=M_{2,p}=1$ by $\m{M}=1$.
\mw{Also, let}
\begin{itemize}
\item $\mathcal{E}_0$ be the event that there is no pair $(\ell_1, \ell_2)\in[2^{nR_1'}]\times [2^{nR_2'}]$ satisfying
\begin{equation*}
(U_0^n(1), U_1^{n}(1,1,\ell_1) , U_2^{n}(1,1,\ell_2)) \in \mathcal{T}_{\eps/32}^{(n)}(P_{U_0U_1U_2}).
\end{equation*}
\item $\mathcal{E}_{0i}$ be the event that
\begin{equation*}
(U_0^n(1), U_i^{n}(1,1,L_i^*),Y_i^n ) \notin \mathcal{T}_{\eps}^{(n)}(P_{U_0U_iY_i}),
\end{equation*}
where $L_1^*$ and $L_2^*$ denote the pair of indices chosen during the encoding step.
\item $\mathcal{E}_{1i}$ be the event  that there is a $\hat{m}_c\neq 1$ such that
\begin{equation*}
(U_0^n(\hat{m}_c), U_i^{n}(\hat{m}_c, 1 , L_i^*),Y_i^n ) \in \mathcal{T}_{\eps}^{(n)}(P_{U_0U_iY_i}).
\end{equation*}
\item $\mathcal{E}_{2i}$ be the event that  there is a pair $\hat{m}_i\neq 1$ and $\hat{\ell}_i$ such that
\begin{equation*}
(U_0^n(1), U_i^{n}(1,\hat{m}_i, \hat{\ell}_i),Y_i^n ) \in \mathcal{T}_{\eps}^{(n)}(P_{U_0U_iY_i}).
\end{equation*}
\item $\mathcal{E}_{3i}$ be the event that  there is a tuple $\hat{m}_c\neq 1$,  $\hat{m}_i\neq 1$, and $\hat{\ell}_i$ such that
\begin{equation*}
(U_0^n(\hat{m}_c), U_i^{n}(\hat{m}_c,\hat{m}_i,\hat{\ell}_i),Y_i^n ) \in \mathcal{T}_{\eps}^{(n)}(P_{U_0U_iY_i}).
\end{equation*}
\end{itemize}

When the event $(\m{E}_{0}^{c} \cap \m{E}_{0,i}^{c}  \cap \m{E}_{1,i}^{c}  \cap \m{E}_{2,i}^{c}  \cap \m{E}_{3,i}^{c})$ occurs, then Receiver~$i\in\{1,2\}$ correctly decodes its desired messages $M_0$ and $M_i$. Therefore,  \begin{IEEEeqnarray*}{rCl}
\lefteqn{
\Prv{\textnormal{error}|\m{M}=1} }\\ & \leq  &
\textnormal{Pr} \Bigg( \mathcal{E}_0 \cup \bigg( \bigcup_{i=1}^2 \bigcup_{j=1}^{4} \m{E}_{j,i} \bigg)
 \bigg|\m{M}=1\Bigg)  \nonumber \\  & \leq  &
\Prv{ \mathcal{E}_0 |\m{M}=1}  \nonumber \\ &&+ \sum_{ i=1}^{2} \big( \Prv{ \mathcal{E}_{0i}|{\mathcal{E}}^{c}_0,\m{M}=1} + \Prv{ \mathcal{E}_{1i}|{\mathcal{E}}_{0i}^{c},\m{M}=1} \nonumber\\
& & \quad \qquad + \Prv{ \mathcal{E}_{2i}|{\mathcal{E}}_{0i}^{c},\m{M}=1}   + \Prv{ \mathcal{E}_{3i}|{\mathcal{E}}_{0i}^{c},\m{M}=1} \big).
\end{IEEEeqnarray*}
We consider each of the terms separately. A nonnegative function $\delta(\eps)$ satisfying $\delta(\eps)\to 0$ as $\eps \to 0$ can be chosen such that the following statements hold.
\begin{itemize}
\item By the code construction and by a conditional version of the covering lemma (Lemma~\ref{lem:covering}),
\begin{equation} \label{eq:l1}
\lim_{n\to 0} \Prv{ \mathcal{E}_0|\m{M}=1 } =0,
\end{equation}
whenever
\begin{equation}\label{eq:M1}
R_1'+R_2'>I(U_1;U_2|U_0)+\delta(\eps).
\end{equation}
\item Since the channel outputs $Y_i^n$ is a $P_{Y_i|X}$-i.i.d. sequence given $X^n$ and by the conditional typicality lemma (Lemma~\ref{lem:cond_typ}),
\begin{equation}\label{eq:channel_cond}
\lim_{n\to 0}  \Prv{ \mathcal{E}_{0i}|{\mathcal{E}}_0^{c},\m{M}=1} =0.
\end{equation}
\item By the code construction and by the packing lemma (Lemma~\ref{lem:packing}),
\begin{equation}\label{eq:l2}
\lim_{n\to 0}  \Prv{ \mathcal{E}_{1i}|{\mathcal{E}}_{0i}^{c},\m{M}=1}  =0,
\end{equation}
whenever
\begin{equation}\label{eq:M2}
R_{0}+R_{1,c}+R_{2,c} < I(U_0, U_i;Y_i) -\delta(\eps).
\end{equation}
\item By the code construction and by the packing lemma:
\begin{equation}\label{eq:l3}
\lim_{n\to 0} \Prv{ \mathcal{E}_{2i}|{\mathcal{E}}_{0i}^{c},\m{M}=1}=0,
\end{equation}
whenever
\begin{equation}\label{eq:M3}
R_{1,p}+R_i'< I(U_i;Y_i|U_0) -\delta(\eps).
\end{equation}

\item Again,  by the code construction and by the packing lemma:
\begin{equation}\label{eq:l4}
\lim_{n\to 0}  \Prv{ \mathcal{E}_{3i}|{\mathcal{E}}_{0i}^{c},\m{M}=1}  =0,
\end{equation}
whenever
\begin{equation} \label{eq:M4}
R_{0}+R_{1,c}+R_{2,c}+R_{i,p}+R_i' < I(U_0, U_i;Y_i) -\delta(\eps).
\end{equation}
 \end{itemize}

 Thus, we conclude that if for $i\in\{1,2\}$
 \begin{subequations} \label{eq:M_befor}
 \begin{IEEEeqnarray}{rCl}
 R_1'+R_2' &>& I(U_1;U_2|U_0)+\delta(\eps)  \\
 R_{i,p}+R_i'& <&  I(U_i;Y_i|U_0) -\delta(\eps)  \\
 R_{0}+R_{1,c}+R_{2,c}+R_{i,p}+R_i' & <&  I(U_0, U_i;Y_i) -\delta(\eps),\IEEEeqnarraynumspace
 \end{IEEEeqnarray}
 \end{subequations}
then the average (over random codebooks, messages, and channel realizations) probability of error of the described scheme  tends to 0 as the blocklength $n$ tends to infinity.
The existence of a deterministic   scheme with average (over messages and channel realizations) probability of error tending to 0 as $n$ tends to infinity  follows then from standard arguments.

By the Fourier-Motzkin elimination algorithm we conclude that whenever
\begin{equation}\label{eq:conditionMarton}
I(U_1;Y_1|U_0) + I(U_2;Y_2|U_0) \geq I(U_1;U_2|U_0)
\end{equation}
then for every rate tuple $(R_0,R_1,R_2)$ satisfying
\begin{subequations}\label{eq:r2}\label{eq:M_after}
 \begin{IEEEeqnarray}{rCl}
R_0+R_1 &<& I(U_0,U_1;Y_1)-\delta(\eps) \\
R_0+R_2 &<& I(U_0,U_2;Y_2)-\delta(\eps) \\
R_{0}+R_{1}+R_2&<& I(U_1;Y_1|U_0)+I(U_2;Y_2|U_0) \nonumber \\
  & &+ \min_{i=1,2} I( U_0;Y_i) - I(U_1;U_2|U_0)-\delta(\eps) \nonumber\\\\
2 R_{0}+R_{1}+R_2& <&  I(U_0, U_1;Y_1)+I(U_0,U_2;Y_2)\nonumber \\ &&- I(U_1;U_2|U_0)-\delta(\eps)\label{eq:r2e}
 \end{IEEEeqnarray}
 \end{subequations}
for a suitable $\delta(\eps)\to 0$ as $\eps\to 0$, there exists a choice of the rates $R_{1,p}$, $R_{1,c}$, $R_{2,p}$, $R_{2,c}$, $R_1'$, $R_2'>0$ such that $R_1=R_{1,p}+R_{1,c}$ and $R_2=R_{2,p}+R_{2,c}$ and such that \eqref{eq:M_befor} holds.

Notice that for every choice of $(U_0,U_1,U_2, X)$ that does not satisfy \eqref{eq:conditionMarton} we can strictly enlarge the rate region  \eqref{eq:r2}
 if we replace the  random triple $(U_0,U_1,U_2)$ by $({U}_0', {U}_1', {U}_2')$ where ${U}_1'$ and ${U}_2'$ are constants and ${U}_0'=(U_0,U_1,U_2)$. The new choice $({U}_0', {U}_1', {U}_2', X)$ moreover satisfies \eqref{eq:conditionMarton} because both sides are 0. Also,  $X$ can be written as a function of the new auxiliaries $U_0', U_1', U_2'$.
We thus conclude that the rate region in \eqref{eq:r2}
is achievable also when \eqref{eq:conditionMarton} is violated.

Taking $\eps\to 0$, \mw{now establishes}
the inclusion $\set{R}_{\textnormal{Marton}}\subseteq \set{C}_{\textnormal{NoFB}}$.

The following two remarks are found useful in the analysis \mw{of our feedback scheme in Appendix~\ref{eq:comerror}.}
\begin{remark}\label{rem:typU}
Under conditions  \eqref{eq:conditionMarton} and \eqref{eq:M_after} there exists an associated choice of  parameters for our scheme such that the associated auxiliary codewords satisfy
\begin{align*}
\begin{multlined}
\Pr\!\Big((U_0^{n}(M_c), U_1^n(M_c, M_{1,p}, L_1^*), U_2^{n}(M_c, M_{2,p}, L_2^*))  \\ \hspace{5.8cm} {\in \m{T}_{\eps/32}^{(n)}(P_{U_0U_1U_2})\Big)}\\
\to  1 \quad \textnormal{as} \quad n\to \infty.
\end{multlined}
\end{align*}
\end{remark}

\begin{remark}\label{rem:Marton}
Inspecting the proof, we see that the memoryless channel property has been used only to establish the limit~\eqref{eq:channel_cond}. The other limits ~\eqref{eq:l1}, \eqref{eq:l2}, \eqref{eq:l3}, and \eqref{eq:l4} follow solely from the way we constructed the code. Suppose now we replace the memoryless channel with a general channel $P_{Y^n|X^n}$. Then under conditions  \eqref{eq:conditionMarton} and \eqref{eq:M_after}, there exists an associated choice of parameters for our scheme such that the average error probability goes to zero as $n\to \infty$, if for $i\in\{1,2\}$:
\begin{IEEEeqnarray*}{rCl}\lefteqn{
\Prv{ (U^n_0(M_c), U^n_i(M_c,M_{i,p}, L_i^*), Y_i^n) \in \mathcal{T}_{\eps}^{(n)}(P_{U_0U_iY_i})}}\qquad  \nonumber \\ &&\!\to 1\quad \textnormal{as} \quad n\to \infty.\hspace{8cm}
\end{IEEEeqnarray*}
\end{remark}

\newcommand{\KL}{\m{K}=\m{L}^*={1}}
\section{Analysis of the Lossy Gray Wyner Scheme with Side-Information}\label{sec:GWSI_analysis}
We analyze the failure probability $\Prv{\m{E}^{(1)}\cup \m{E}^{(2)}}$ associated with \mw{our} random coding scheme, where $\m{E}^{(i)}$ is the event that Receiver~$i$ fails, i.e., $(X^{n},\hat{V}_i^{n})\not\in\m{T}^{n}_ \eps(P_{XV_i})$.

Let $K_{0,0}^*$, $K_{1,0}^*$, $K_{2,0}^*$, $L_{1,0}^*$, $L_{2,0}^*$, $K_{0,1}^*$, $K_{1,1}^*$, $L_1^*$, $K_{0,2}^*$, $K_{2,2}^*$, $L_2^*$ be the tuple of indices chosen by the sender.
Also, let
\begin{itemize}
\item $\m{E}_0$ be the event that $X^n\not\in T^n_{\eps/8}(P_X)$;
\item $\m{E}_1$ be the event that
\begin{IEEEeqnarray*}{lCl}\lefteqn{\forall k_{0,0}, k_{1,0}, \ell_{1,0}:} \nonumber \\ &&\big(X^n,V_0^n(1; k_{0,0},k_{1,0}, \ell_{1,0})\big)\not\in T^n_{\eps/4}(P_{XV_0});\hspace{1cm}\end{IEEEeqnarray*}
\item $\m{E}_{2,i}$, for $i\in\{1,2\}$,  be the event that
\begin{IEEEeqnarray*}{lCl}\lefteqn{
\forall k_{0,i}, k_{i,i},\ell_i:\nonumber }\\
 &&  \;(X^n,V_0^n(i; K_{0,0}^*, K_{i,0}^*, L_{i,0}^*),V_i^n( k_{0,i},k_{i,i},\ell_i))\nonumber \\ & & \hspace{5.5cm}\not\in T^n_{\eps/2}(P_{XV_0V_i});\hspace{6cm}\end{IEEEeqnarray*}
\item $\m{E}_{3,i}$,  for $i\in\{1,2\}$, be the event that
\begin{IEEEeqnarray*}{lCl}\lefteqn{
(V_0^n(i; K_{0,0}^*, K_{i,0}^{*},L_{i,0}^*),V_i^n( K_{0,i}^*, K_{i,i}^*,L_i^*),Y_i^n)}\nonumber \\ & & \hspace{5.5cm}\not\in T^n_ \eps(P_{V_0V_iY_i});
\end{IEEEeqnarray*}
\item $\m{E}_{4,i}$, for $i\in\{1,2\}$, be the event that
\begin{IEEEeqnarray}{rCl}
\lefteqn{\exists \ell_i\neq L_i^*: }  \nonumber \\
& & (V_0^n(i; K_{0,0}^*, K_{i,0}^*,L_{i,0}^*),{V}_i^n( K_{0,i}^*, K_{i,i}^*,\ell_i),Y_i^n)\nonumber \\ & &  \hspace{5.5cm}\in T^n_ \eps(P_{V_0V_iY_i}); \nonumber
\end{IEEEeqnarray}
\item $\m{E}_{5,i}$,  for $i\in\{1,2\}$,  be the event that
\begin{IEEEeqnarray}{rCl}
\lefteqn{\exists \ell_{i,0}\neq L_{i,0}^*,\ell_i \neq L_i^*:} \quad \nonumber \\
& & (V_0^n(i; K_{0,0}^*, K_{i,0}^*, \ell_{i,0}),{V}_i^n( K_{0,i}^*,K_{i,i}^*,\ell_i),Y_i^n)\nonumber \\ & &  \hspace{5.5cm}\in T^n_ \eps(P_{V_0V_iY_i}).\nonumber
\end{IEEEeqnarray}
\end{itemize}
Notice that whenever event $(\m{E}_0^{c}\cap \m{E}_1^{c}\cap \m{E}_{2,i}^{c})$ occurs, then $(X^n,V_i^n( K_{0,i}^*,K_{i,i}^*,L_i^*))\in T^n_ \eps(P_{XV_i})$. If additionally also event $(\m{E}_{3,i}^{c} \cap \m{E}_{4,i}^{c}\cap \m{E}_{5,i}^{c})$ occurs, then Receiver~$i$ produces $\hat{V}_i^n=V_i^n( K_{0,i}^*,K_{i,i}^*,L_i^*)$. Therefore,
\begin{IEEEeqnarray}{rCl}
\Pr(\m{E}^{(i)})
&\leq &\Pr\left(\m{E}_0\cup\m{E}_1\cup\m{E}_{2,i}\cup\m{E}_{3,i}\cup\m{E}_{4,i}\cup \m{E}_{5,i}\right) \nonumber
\\
&\leq&  \Pr(\m{E}_0)+ \Pr(\m{E}_1|\m{E}_0^c)\nonumber \\ &&+ \Pr(\m{E}_{2,i}|\m{E}_1^c) +  \Pr(\m{E}_{3,i}|\m{E}_{2,i}^c)
 \nonumber \\ &&+ \Pr(\m{E}_{4,i})+ \Pr(\m{E}_{5,i}).\label{eq:errorupper}
\end{IEEEeqnarray}
We analyze each of the summands separately. Hereinafter, a nonnegative function $\delta(\eps)$ satisfying $\delta(\eps)\to 0$ as $\eps \to 0$, can be chosen such that the statements hold.
\begin{itemize}
\item  Since $X^n$ is $P_X$-i.i.d. and by the weak law of large numbers:
\begin{equation}\label{eq:E0GW}
\lim_{n\rightarrow \infty}\Prv{\m{E}_0}= 0.
\end{equation}
\item  By the code construction and the covering lemma (Lemma~\ref{lem:covering}):
\begin{equation}\label{eq:E1GW}
\lim_{n\rightarrow \infty} \Prv{\m{E}_1|\m{E}_0^c}=0
\end{equation}
whenever
\begin{equation}\label{eq:rates0}
\mw{{R}_0'+R_{0,0} > I(X;V_0)+ \delta(\eps)}.
\end{equation}
\item Again, by the code construction and the covering lemma:
\begin{equation}\label{eq:E2GW}
\lim_{n\rightarrow \infty} \Prv{\m{E}_{2,i}|\m{E}_1^c}=0
\end{equation}
whenever
\begin{equation}\label{eq:ratesi}
\mw{R_i'+R_{0,i}+R_{i,i}   > I(V_i;X,V_0)+ \delta(\eps)}.
\end{equation}
\item  The pair $\big(V_0^n(i; K_{0,0}^*, K_{i,0}^*,L_{i,0}^*), V_i^n( K_{0,i}^*,K_{i,i}^*,L_i^*)\big)$ depends on $Y_i^n$ only through  $X^n$, i.e., the Markov chain
\[V_0^n(i; K_{0,0}^*, K_{i,0}^*,L_{i,0}^*), V_i^n( K_{0,i}^*,K_{i,i}^*,L_i^*)  \markov X^n\markov  Y_i^n\]
holds. Therefore, $Y_i^n$ is $P_{Y_i|XV_0V_i}=P_{Y_i|X}$-independent given $(X^n, V_0^n(i; K_{0,0}^*, K_{i,0}^*, L_{i,0}^*), V_i^n( K_{0,i}^*,K_{i,i}^*,L_i^*))$ and by the  conditional typicality lemma (Lemma~\ref{lem:cond_typ}):
\begin{equation}\label{eq:E3GW}
\lim_{n\rightarrow \infty} \Prv{\m{E}_{3,i}|\m{E}_{2,i}^c}=0.
\end{equation}
\item Notice that the codewords $\{V_i^n( K_{0,i}^*,K_{i,i}^*,\ell_i)\}$ for $\ell_i\in[2^{nR_i'}]\backslash\{L_i^*\}$ are not independent and $P_{V_i}$-i.i.d.\footnote{This can be seen with the following simple example. Let the heights of two students $A_0$ and $A_1$ be
uniformly distributed over the interval $[1.7, 1.9]$ m and independent of each other. Also, let $C$ be the index of the student that has height larger than $1.89$m if this index is unique; otherwise let $C$ be Bern($\frac{1}{2}$). Let $\bar{C}$ be the index in $\{0,1\}$ not equal to $C$.
Notice that $\Prv{A_0\geq 1.89}=\frac{1}{20}$, whereas $\Prv{A_{\bar{C}}\geq 1.89}= \Prv{A_0\geq 1.89 \textnormal{ and } A_1 \geq 1.89}=\frac{1}{400}$. Thus, $A_{\bar{C}}$ is not uniform over $[1.7,1.9]$.}
\mw{However, following similar steps as in~\cite[Appendix 12A]{El-Gamal--Kim2009},
one can prove
 Inequality~\eqref{eq:ineqtyp} on top of the next page for arbitrary $(k_{0,i}^*, k_{i,i}^*, \ell_i^*)\in[2^{nR_{0,i}}]\times [2^{nR_{i,i}}]\times [2^{nR_i'}]$
\begin{figure*}
\begin{IEEEeqnarray}{rCl}
\Prv{ \m{E}_{4,i}  }& = &
\label{eq:in01}
\lefteqn{   \Prv{ \bigcup_{\substack{\ell_i \in [2^{nR_i'}] \\ \ell_i \neq L_i^*}} (V_0^n(i;K_{0,0}^*,K_{i,0}^*,L_{i,0}^*), V_i^n(K_{0,i}^*, K_{i,i}^*, \ell_i) , Y_i^n) \in T_\eps^n(P_{V_0V_iY_i}) }  }\qquad \nonumber\\ & = &
 \Prv{ \bigcup_{\ell_i=2}^{\lfloor2^{nR_i'}\rfloor} (V_0^n(i;K_{0,0}^*,K_{i,0}^*,L_{i,0}^*), V_i^n(1,1, \ell_i) , Y_i^n) \in T_\eps^n(P_{V_0V_iY_i})\Bigg|K_{0,i}^*=K_{i,i}^*= L_i^*=1}\\
& \leq &
 \Prv{ \bigcup_{\ell_i=1}^{\l2^{nR_i'}\r}  (V_0^n(i;K_{0,0}^*,K_{i,0}^*,L_{i,0}^*), V_i^n(1,1, \ell_i) , Y_i^n)\in T_\eps^n(P_{V_0V_iY_i})\Bigg|K_{0,i}^*=k_{0,i}^*, K_{i,i}^*=k_{i,i}^*, L_i^*=\ell_i^*}
\label{eq:ineqtyp}
\end{IEEEeqnarray}
\hrulefill
\end{figure*}
which directly yields Inequality~\eqref{eq:in0a}, also shown on the next page.
\begin{figure*}
\begin{IEEEeqnarray}{rCl}\label{eq:in0a}
\Prv{ \m{E}_{4,i}  }& \leq &
\Prv{ \bigcup_{\ell_i=1}^{\lfloor 2^{nR_i'}\rfloor} (V_0^n(i; K_{0,0}^*, K_{i,0}^{*},L_{i,0}^*),{V}_i^n(1,1,\ell_i),Y_i^n)\in T^n_\eps(P_{V_0V_iY_i})}
\end{IEEEeqnarray}
\hrulefill
\end{figure*}
Here, Equality~\eqref{eq:in01} holds by the symmetry of the code construction. For $(k_{0,i}^*,k_{i,i}^*, \ell_i^*)=(1,1,1)$, Inequality~\eqref{eq:ineqtyp} is straightforward; for  $(k_{0,i}^*,k_{i,i}^*)=(1,1)$ and $\ell_i^*>1$ it follows by this first case and the symmetry of the code construction; and for $(k_{0,i}^*,k_{i,i}^*)\neq(1,1)$ and $\ell_i^*$ arbitrary it follows by  \eqref{eq:in1}--\eqref{eq:in11} which hold again by the  symmetry of the code construction and because conditioned on  $K_{0,i}^*=K_{i,i}^*=L_i^*=1$ every set of $\l2^{nR_i'}\r-1$ codewords $\{V_i^n(k_{0,i},k_{i,i},\ell_i)\}$ for $(k_{0,i}, k_{i,i},\ell_i)\neq (1,1,1)$ has the same joint distribution.
\begin{figure*}
\begin{IEEEeqnarray}{rCl}
\lefteqn{   \Prv{ \bigcup_{\ell_i=2}^{\l2^{nR_i'}\r}V_0^n(i;K_{0,0}^*,K_{i,0}^*,L_{i,0}^*), V_i^n(1,1, \ell_i) , Y_i^n) \in T_\eps^n(P_{V_0V_iY_i}) \Bigg|  K_{0,i}^*=K_{i,i}^*= L_i^*=1}  }\qquad \nonumber \\
& = &
 \Prv{ \bigcup_{\ell_i=2}^{\l2^{nR_i'}\r} (V_0^n(i;K_{0,0}^*,K_{i,0}^*,L_{i,0}^*), V_i^n(1,1, \ell_i) , Y_i^n) \in T_\eps^n(P_{V_0V_iY_i}) \Bigg| K_{0,i}^*= k_{0,i}^*, K_{i,i}^*=k_{i,i}^*, L_i^*=\ell_i^*}\label{eq:in1}\\ & \leq  &
 \Prv{ \bigcup_{\ell_i=1}^{\l2^{nR_i'}\r} (V_0^n(i;K_{0,0}^*,K_{i,0}^*,L_{i,0}^*), V_i^n(1,1, \ell_i) , Y_i^n) \in T_\eps^n(P_{V_0V_iY_i}) \Bigg| K_{0,i}^*= k_{0,i}^*, K_{i,i}^*=k_{i,i}^*, L_i^*=\ell_i^*  }\label{eq:in11}
\end{IEEEeqnarray}
\hrulefill
\end{figure*}
}

Notice that on the right-hand side of \eqref{eq:in0a} we have the probability that one of the $\l2^{nR_i'}\r$ independent and $P_{V_i}$-i.i.d.  codewords $\{{V}_i^n(1,1,\ell_i)\}_{\ell_i=1}^{\lfloor 2^{nR_i'}\rfloor}$ is jointly $\eps$-typical with the pair $(V_0^n(i; K_{0,0}^*, K_{i,0}^{*},L_{i,0}^*), Y_i^n)$.
Thus, by the packing lemma (Lemma~\ref{lem:packing}) the probability on the right-hand side of \eqref{eq:in0a} tends to 0 as $n$ tends to $\infty$ whenever
\begin{equation}\label{eq:rate1p}
R_i' < I(V_i;V_0, Y_i)-\delta(\eps).
\end{equation}

We thus conclude that
\begin{equation}\label{eq:E4GW}
\lim_{n\rightarrow \infty} \Prv{\m{E}_{4,i}}=0
\end{equation}
whenever \eqref{eq:rate1p} holds.
\item Following similar steps \mw{as above, we can prove upper bound \eqref{eq:ineqtyp5}.}
\begin{figure*}
\begin{IEEEeqnarray}{rCl}\label{eq:ineqtyp5}
\Prv{ \m{E}_{5,i}  }&\leq& \Prv{ \bigcup_{\substack{\ell_{i,0}\in \left[2^{n({R}'_0-R_{i,0})}\right], \\ \ell_i\in[ 2^{nR_i'}]
}
} ( {V}_0^n(i; 1,1,\ell_{i,0}),{V}_i^n(1, 1, \ell_i),Y_i^n)\in T^n_\eps(P_{V_0V_iY_i})}
\end{IEEEeqnarray}
\hrulefill
\end{figure*} Then, by the multivariate packing lemma (Lemma~\ref{lem:mv_packing}):
\begin{equation}\label{eq:E5GW}
\lim_{n\rightarrow \infty} \Prv{\m{E}_{5,i}}=0,
\end{equation}
whenever
\begin{equation}\label{eq:rate10p}
{R}_0' - R_{i,0} +R_i' < I(V_0;Y_i)+I(V_i;V_0, Y_i)- \delta(\eps).
\end{equation}
\end{itemize}
Combining \eqref{eq:errorupper} with \eqref{eq:rates0}, \eqref{eq:ratesi},  \eqref{eq:rate1p}, and \eqref{eq:rate10p} we obtain that $\Prv{\m{E}^{(1)}}$ and
$\Prv{\m{E}^{(2)}}$ both tend to $0$ as $n\rightarrow\infty$ whenever:
\begin{subequations}\label{eq:GW_before}
\begin{IEEEeqnarray}{rCl}
{R}_0'+R_{0,0}&>& I(X;V_0)+\delta(\eps)  \\
R_1'+R_{0,1}+R_{1,1} &>& I(V_1;X,V_0)+\delta(\eps) \\
R_2'+R_{0,2}+R_{2,2} &>& I(V_2;X,V_0) +\delta(\eps)\\
{R}_0'-R_{1,0}+R_1' &<& I(V_0; Y_1) + I(V_1; V_0,Y_1)-\delta(\eps) \IEEEeqnarraynumspace\\
{R}_0' - R_{2,0}+R_2' &<& I(V_0; Y_2) + I(V_2; V_0,Y_2) -\delta(\eps) \\
R_1' &<& I(V_1; V_0,Y_1) -\delta(\eps)\\
R_2' &<& I(V_2; V_0,Y_2)-\delta(\eps).
\end{IEEEeqnarray}
\end{subequations}

We now argue that with an appropriate choice of the auxiliary rates $R_{0}'$, $R_{1}'$, $R_{2}'$, $R_{0,0}$, $R_{0,1}$, $R_{0,2}$, $R_{1,0}$, $R_{1,1}$, $R_{2,0}$, $R_{2,2}>0$ our scheme achieves the region ${\set{R}}_{\tn{LGW}}^{\textnormal{inner}}$. We first replace $R_{i,i}$ by $R_{i}-R_{i,0}$, for $i\in\{1,2\}$ and $R_{0,0}$ by $R_0-R_{0,1}-R_{0,2}$ to obtain
\begin{subequations}\label{eq:GW_before2}
\begin{IEEEeqnarray}{rCl}
{R}_0'+R_{0}-R_{0,1}-R_{0,2}&>& I(X;V_0)+\delta(\eps)  \\
R_1'+R_{0,1}+R_{1}- R_{1,0}&>& I(V_1;X,V_0)+\delta(\eps) \\
R_2'+R_{0,2}+R_{2}-R_{2,0} &>& I(V_2;X,V_0) +\delta(\eps)\\
{R}_0'-R_{1,0}+R_1' &<& I(V_0; Y_1) + I(V_1; V_0,Y_1)-\delta(\eps)\nonumber\\ \\
{R}_0' - R_{2,0}+R_2' &<& I(V_0; Y_2) + I(V_2; V_0,Y_2) -\delta(\eps)\nonumber\\ \\
R_1' &<& I(V_1; V_0,Y_1) -\delta(\eps)\\
R_2' &<& I(V_2; V_0,Y_2)-\delta(\eps).
\end{IEEEeqnarray}
\end{subequations}
Then,  employing the Fourier-Motzkin elimination algorithm to eliminate the nuisance variables $R_{0}'$, $R_{1}'$, $R_{2}'$, $R_{0,1}$, $R_{0,2}$, $R_{1,0}$, $R_{2,0}$, we obtain that if $(R_0, R_1, R_2)$ satisfies
\begin{subequations}\label{eq:LGW_eps0}
\begin{IEEEeqnarray}{rClCl}
 R_0 + R_1 &>& I(X;V_0) + I(V_1;X,V_0) - I(V_0; Y_1) \nonumber\\ &&  - I(V_1; V_0,Y_1) +\delta(\eps)\\
R_0+R_2 &>& I(X;V_0) + I(V_2;X,V_0) - I(V_0; Y_2) \nonumber \\ & &- I(V_2; V_0,Y_2)+\delta(\eps)\\
R_0+R_1+R_2 & >& I(X;V_0) +I(V_1;X,V_0) +I(V_2;X_,V_0) \nonumber \\ & &- I(V_1;V_0,Y_1)-I(V_2;V_0,Y_2) \nonumber \\ &&- \min_{i} I(V_0;Y_i)   + \delta(\epsilon)
\end{IEEEeqnarray}
\end{subequations}
then there exists a choice of nonnegative rates ${R}_0'$, $R_1'$, $R_2'$, $R_{0,1}$, $R_{0,2}$, $R_{1,0}$, $R_{2,0}$ that satisfies \eqref{eq:GW_before2} and
\begin{IEEEeqnarray*}{rCl}
R_{1}-R_{1,0}&\geq& 0\\
R_2-R_{2,2}&\geq &0\\
{R}_0' - R_{1,0}&\geq& 0\\
{R}_{0}'-R_{2,0}&\geq &0\\
R_0-R_{0,1}-R_{0,2}& \geq & 0.
\end{IEEEeqnarray*}
\mw{Due to the Markov chain $(V_0,V_1,V_2) \markov X \markov (Y_1,Y_2)$ the constraints in \eqref{eq:LGW_eps0} are equivalent to}
\begin{subequations} \label{eq:LGW_eps}
\begin{IEEEeqnarray}{rClCl}
R_0+ R_1 &>&I(X;V_0,V_1|Y_1) +\delta(\eps) \\
 R_0+R_2 &>&I(X;V_0,V_2|Y_2) +\delta(\eps)\\
 R_0+R_1+R_2 & >& I(X;V_1|V_0, Y_1)+I(X;V_2|V_0,Y_2)\nonumber \\
&& + \max_{i} I(X;V_0|Y_i) + \delta(\epsilon).
\end{IEEEeqnarray}\label{eq:GW_after}
\end{subequations}
Thus, we conclude that the region \eqref{eq:GW_after} is $\eps$-achievable for all choices  of the auxiliary random variable $V_0$ satisfying the Markov chain  $(V_0,V_1,V_2) \markov X \markov (Y_1,Y_2)$. Letting $\eps\to 0$, the achievability of ${\set{R}}_{\tn{LGW}}^{\textnormal{inner}}$ is established.

The following remark is found useful in the analysis \mw{of the feedback scheme in Appendix~\ref{eq:comerror}.}
\begin{remark}\label{rem:LGW}
In our error analysis, only Limits \eqref{eq:E0GW} and \eqref{eq:E3GW}  rely on the assumption that $(X^n, Y_1^n, Y_2^n)$ are $P_{XY_1Y_2}$-i.i.d. It is easy to check that replacing this assumption with the more general assumptions
\begin{enumerate}[(i)]
\item $\Pr(X^n\in \m{T}_{\eps/8}^n(P_X)) \to 1$ as $n\to\infty$.
\item $(Y_1^n,Y_2^n)$ is $P_{Y_1Y_2|X}$-independent given $X^n$.
\end{enumerate}
still guarantees the existence of associated parameters such that the scheme above $\eps$-achieves the region (\ref{eq:LGW_eps}). In particular,
\begin{IEEEeqnarray*}{rCl}
\Prv{ (X^n,V_i^n(K_{0,i}^*,K_{i,i}^*,L_i^*))\notin T^n_ \epsilon(P_{XV_i}) } \to\ 0
\end{IEEEeqnarray*}
and
\begin{IEEEeqnarray*}{rCl}
\Prv{\hat{V}_i^n \neq V_i^n(K_{0,i}^*, K_{i,i}^*,L_i^*)}\to 0,
\end{IEEEeqnarray*}
for $i\in\{1,2\}$, as $n\to\infty$.
\end{remark}

\section{Convexity in Theorem \ref{thrm:LGW}}\label{sec:app_conv}
Let $\{V_{0,j},V_{1,j},V_{2,j},X_j,Y_{1,j},Y_{2,j}\}_{j\in\{0,1\}}$ be two sets of mutually independent random variables for $j\in\{1,2\}$, where
\begin{itemize}
\item  $(X_j,Y_{1,j},Y_{2,j})\sim P_{XY_1Y_2}$;
\item $(V_{0,j},V_{1,j},V_{2,j}) \markov X_j \markov (Y_{1,j},Y_{2,j})$.
\item $P_{V_{i,j}|X_j} = P_{V_i|X}$ for $i\in\{1,2\}$.
\end{itemize}

Let $Q\sim \tn{Bern}(\alpha)$ be independent of the union of the two sets, and define $\bar{V}_0 \dfn V_{0,Q}, \bar{V}_i \dfn V_{i,Q}, \bar{X} \dfn X_{Q}, \bar{Y}_i \dfn Y_{i,Q}$, for $i\in\{1,2\}$.
Notice that as the law of $(X_1,Y_{1,1}, Y_{2,1})$ and the law of $(X_2, Y_{1,2}, Y_{2,2})$ are the same, the "time-sharing" random variable $Q$ is independent of the triplet $(\bar{X},\bar{Y}_1,\bar{Y}_2)$. Therefore, and since by assumption
\begin{equation*}
 (\bar{V}_{0},\bar{V}_{1},\bar{V}_{2}) \markov (\bar{X }, Q)\markov (\bar{Y}_{1},\bar{Y}_{2}),
\end{equation*}
we  conclude that  defining $\tilde{V}_0\dfn(Q,\bar{V}_0)$ we have the Markov chain
\begin{equation}\label{eq:M}
(\tilde{V}_{0},\bar{V}_{1},\bar{V}_{2}) \markov \bar{X }\markov (\bar{Y}_{1},\bar{Y}_{2}).
\end{equation}

We further notice that for $i\in\{1,2\}$:
\begin{IEEEeqnarray}{rCl}\label{eq:conv1}
I(\bar{X}; \bar{V}_i| \bar{V}_0, \bar{Y}_i, Q ) = I(\bar{X}; \bar{V}_i| \tilde{V}_0, \bar{Y}_i)
\end{IEEEeqnarray}
and
\begin{IEEEeqnarray}{rCl}\label{eq:conv2}
I(\bar{X}; \bar{V}_0| \bar{Y}_i, Q ) = I(\bar{X}; \bar{V}_0, Q |\bar{Y}_i) = I(\bar{X}; \tilde{V}_0|\bar{Y}_i) ,
\end{IEEEeqnarray}
where the first equality holds because of the independence of $Q$ and $(\bar{X},\bar{Y}_i)$. Moreover, by \eqref{eq:conv1} and \eqref{eq:conv2}
\begin{equation}
I(\bar{X}; \bar{V}_0, \bar{V}_i |\bar{Y}_i, Q) =  I(\bar{X}; \bar{V}_i, \tilde{V}_0| \bar{Y}_i).
\end{equation}
Combining these inequalities with the Markov condition,  we conclude that the region   $\set{R}_{\textnormal{LGW}}^{\textnormal{inner}}$ is convex.

\section{Error Analysis for the Feedback Scheme}\label{eq:comerror}
We bound the average probability of error (where the average is over the random messages, codes, and channel realizations).
Let $\m{E}$ be the error event:
\begin{equation*}
\m{E} \dfn \bigcup_{i=1}^2 \bigcup_{b=1}^{B}\left\{(\wh{M}_{0,i,(b)},\wh{M}_{i,(b)})\neq \left(M_{0,(b)},M_{i,(b)}\right)\right\}.
\end{equation*}
Moreover, for each $b\in[B+1]$, let $\m{F}_b$ be the error event of the Marton code in block $b$:
\begin{equation*}
\m{F}_b \dfn \bigcup_{i=1}^2 \left\{(\wh{J}_{0,i,(b)},\wh{J}_{i,(b)})\neq \left(J_{0,(b)},J_{i,(b)}\right)\right\}.
\end{equation*}
Then,
\begin{equation*}
\Pr(\m{E}) \leq \Pr\left( \bigcup_{b=1}^{B+1}\m{F}_b \right) \leq \sum_{b=1}^{B} \Pr(\m{F}_b| \m{F}_{b+1}^{\textnormal{c}}) + \Pr(\m{F}_{B+1}).
\end{equation*}

By construction, we have that $\Pr(\m{F}_{B+1})\to 0$  as $n\to \infty$. Let us now analyze the probability $\Pr(\m{F}_b | {\m{F}}_{b+1}^{\textnormal{c}})$ for a fixed $b\in[B]$. In light of Remark~\ref{rem:Marton}, we see that if
\begin{equation}\label{eq:conditionMarton2}
I(U_1;Y_1,V_1|U_0) + I(U_2;Y_2,V_2|U_0) \geq I(U_1;U_2|U_0);
\end{equation}
and
\begin{subequations}\label{eq:Martoninproof}
\begin{IEEEeqnarray}{rCl}
\bar{R}_0+\bar{R}_1 &<& I(U_0,U_1;Y_1,V_1) - \delta(\eps)
\\
\bar{R}_0+\bar{R}_2 &<& I(U_0,U_2;Y_2,V_2) - \delta(\eps)
\\
\nonumber \bar{R}_0+\bar{R}_1+\bar{R}_2 &<& I(U_1; Y_1,V_1|U_0) + I(U_2; Y_2,V_2|U_0)   \\
 &&+ \min_{i}I(U_0;Y_i,V_i)- I(U_1;U_2|U_0)- \delta(\eps);\nonumber \\\\
 \nonumber 2\bar{R}_0+\bar{R}_1+\bar{R}_2 &<& I(U_0,U_1; Y_1,V_1) + I(U_0,U_2; Y_2,V_2)   \\
 &&- I(U_1;U_2|U_0)- \delta(\eps);
\end{IEEEeqnarray}
\end{subequations}
and  for $i\in\{1,2\}$:
\begin{equation} \label{eq:whattoprove}
\Pr((U_{0,(b)}^n, U_{i,(b)}^n, Y_{i,(b)}^n, \hat{V}_{i,(b)}^n) \not\in\m{T}^n_\eps(P_{U_0U_i Y_i, V_i})) \to  0
\end{equation}
as $n\to\infty$, then there exists a choice of the parameters such that $\Prv{F_{b}|{F}_{b+1}^{\textnormal{c}}}\to 0$ as $n\to \infty$.

From this point forward we assume that Conditions~\eqref{eq:conditionMarton2} and (\ref{eq:Martoninproof}) hold, \mw{and}
prove that if additionally
\begin{subequations}\label{eq:GW_inproof}
\begin{IEEEeqnarray}{rCl}
 \wt{R}_0 + \wt{R}_1 &>& I(U_0,U_1,U_2,\wt{Y};V_0,V_1|Y_1)+ \delta(\eps)\\
\wt{R}_0+\wt{R}_2 &>& I(U_0,U_1,U_2,\wt{Y};V_0,V_2|Y_2)+ \delta(\eps)\\
\wt{R}_0+ \wt{R}_1+\wt{R}_2 &>& I(U_0,U_1,U_2,\wt{Y};V_1|V_0,Y_1) \nonumber \\ &&+ I(U_0,U_1,U_2,\wt{Y};V_2|V_0,Y_2)\nonumber \\&&+\max_{i}  I(U_0,U_1,U_2,\wt{Y};V_0|Y_i) + \delta(\eps)\IEEEeqnarraynumspace
\end{IEEEeqnarray}
\end{subequations}
then the limit \eqref{eq:whattoprove} holds. We notice that
\begin{align}
\nonumber&\Pr\!\left(\!(U_{0,(b)}^n,U_{i,(b)}^n,Y_{i,(b)}^n,\hat{V}_{i,(b)}^n)\not\in\m{T}^n_\eps(P_{U_0U_iY_iV_i})\!\right) \\
\nonumber&\quad \leq \Pr\!\left(\!(U_{0,(b)}^n,U_{i,(b)}^n,Y_{i,(b)}^n,V_{i,(b)}^n)\not\in\m{T}^n_\eps(P_{U_0U_iY_iV_i})\!\right)
\\
&\quad \qquad + \Pr\!\left(\wh{V}_{i,(b)}^n\neq V_{i,(b)}^n\right), \label{eq:whattoprove_exp}
\end{align}
where $V_{1,(b)}^n$ and $V_{2,(b)}^n$ denote the codewords chosen by the LGW-SI encoding rule $\lambda_{(b)}$. We now verify that under conditions (\ref{eq:GW_inproof}), both terms on the right-hand side of (\ref{eq:whattoprove_exp}) vanish as $n\to\infty$.

\mw{Since the input $X_{(b)}^n$ is a component-wise function of $(U_{0,(b)}^n,U_{1,(b)}^n,U_{2,(b)}^n)$
and the channel is memoryless, $(Y_{1,(b)}^{n}, Y_{2,(b)}^{n}, \wt{Y}_{(b)}^n)$ is $P_{Y_1Y_2\wt{Y}|U_0U_1U_2}$-independent given $(U_{0,(b)}^n,U_{1,(b)}^n,U_{2,(b)}^n)$.}  Furthermore, from Marton's code construction and in light of Remark~\ref{rem:typU}, we have that under conditions \eqref{eq:conditionMarton2} and \eqref{eq:Martoninproof}
\begin{IEEEeqnarray*}{rCl}
\Prv{(U_{0,(b)}^n,U_{1,(b)}^n,U_{2,(b)}^n) \notin \m{T}_{\eps/32}(P_{U_0U_1U_2}) } \rightarrow 0.
\end{IEEEeqnarray*}
Therefore, by the conditional typicality Lemma, also
\begin{IEEEeqnarray*}{rCl}
\Prv{(U_{0,(b)}^n,U_{1,(b)}^n,U_{2,(b)}^n,\wt{Y}_{(b)}^n) \notin \m{T}_{\eps/16}(P_{U_0U_1U_2\wt{Y}})} \rightarrow 0\IEEEeqnarraynumspace
\end{IEEEeqnarray*}
as $n\to\infty$.

\mw{Thus, by Remark~\ref{rem:LGW}
(recall we have used the parameter $\eps/2$ for the LGW-SI code) and
under conditions (\ref{eq:GW_inproof})}
\begin{IEEEeqnarray}{lCl}
&\Pr\!\left(\wh{V}_{i,(b)}^n\neq V_{i,(b)}^n\right) \to 0\label{eq:1st_summand}\\
&\Pr\!\left(\!(U_{0,(b)}^n,U_{1,(b)}^n,U_{2,(b)}^n,\wt{Y}_{(b)}^n, V_{i,(b)}^n)\not\in\m{T}^n_{\eps/2}(P_{U_0U_1U_2\wt{Y}V_i})\!\right) \!\!\to \!0\nonumber \label{}\\
\label{eq:2nd_summand}
\end{IEEEeqnarray}
as $n\to\infty$.

\mw{Now, since }
$Y_{i,(b)}^n$ is $P_{Y_i|U_0U_1U_2\wt{Y}}$-independent given $(\!U_{0,(b)}^n,U_{1,(b)}^n,U_{2,(b)}^n,\!\wt{Y}_{(b)}^n)$, and the Markov condition
\[V_{i,(b)}^n\markov (\!U_{0,(b)}^n,U_{1,(b)}^n,U_{2,(b)}^n,\!\wt{Y}_{(b)}^n)\markov Y_{i,(b)}^n\]
holds,
\mw{(\ref{eq:2nd_summand}) and the conditional typicality Lemma imply that}
\begin{IEEEeqnarray}{lCl}\label{eq:tp11}
\Pr\!\left(\!(U_{0,(b)}^n,U_{i,(b)}^n,Y_{i,(b)}^n, V_{i,(b)}^n)\not\in\m{T}^n_\eps(P_{U_0U_iY_iV_i})\!\right) \to 0.\IEEEeqnarraynumspace
\end{IEEEeqnarray}
With (\ref{eq:whattoprove_exp}) and (\ref{eq:1st_summand})
this establishes
(\ref{eq:whattoprove}).
\mw{We thus proved that whenever~\eqref{eq:conditionMarton2}, \eqref{eq:Martoninproof}, and \eqref{eq:GW_inproof} are satisfied, then the probability of error tends to 0 as $n\to \infty$, for any $\eps$ small enough.

Employing the Fourier-Motzkin elimination algorithm on constraints~\eqref{eq:Martoninproof} and (\ref{eq:GW_inproof}) where we replaced $\bar{R}_i$ by $R_i+\tilde{R}_i$, and letting $\epsilon$ tend to 0},
we obtain that under the set of constraints \eqref{eq:inner} and when \eqref{eq:conditionMarton2} holds, then there exists a choice of the parameters such that  the probability of error of our scheme tends to 0 as $n\to\infty$.
Notice however, that when a triplet \mw{$(U_0,U_1,U_2)$ does not satisfy~\eqref{eq:conditionMarton2},}
then the rate region \eqref{eq:inner} is strictly enlarged if \mw{we replace this triplet by }
$({U}_0', {U}_1', {U}_2')$ where ${U}_1'$ and ${U}_2'$ are constants and  ${U}_0'=(U_0,U_1,U_2)$. The new choice $({U}_0', {U}_1', {U}_2')$ moreover satisfies \eqref{eq:conditionMarton2} because both sides are 0. It also satisfies the Markov chain \eqref{eq:M3} and $X$ can be expressed as a function of the new auxiliaries $U_0', U_1', U_2'$. \mw{We can thus ignore constraint~\eqref{eq:conditionMarton2} in the statement of the achievable region.}

We conclude that since by \eqref{eq:effrates} the effective rates of transmission tend to $(R_0, R_1, R_2)$ as $B\to \infty$,  any rate triplet satisfying the constraints \eqref{eq:inner} is achievable by our scheme.

\begin{biographynophoto}{Ofer Shayevitz}
Ofer Shayevitz received the B.Sc. degree (summa cum laude) from the Technion Institute of Technology, Haifa, Israel, in 1997 and the M.Sc. and Ph.D. degrees from the Tel-Aviv University, Tel Aviv, Israel, in 2004 and 2009, respectively, all in electrical engineering.

He is currently a quantitative analyst with the D. E. Shaw group. Before that, he was a Postdoctoral Fellow in the Information Theory and Applications (ITA) Center at the University of California, San Diego. Prior to his graduate studies, he served as an Engineer and Team Leader in the Israeli Defense Forces (1997-2003), and as an Algorithms Engineer with CellGuide, a high-tech company developing low-power GPS navigation solutions (2003-2004). Dr. Shayevitz is the recipient of the ITA Postdoctoral Fellowship (2009-2011), the Adams Fellowship (2006-2008) awarded by the Israel Academy of Sciences and Humanities, the Advanced Communication Center (ACC) Feder Family award for an outstanding Ph.D. thesis (2009), and the Weinstein Prize (2006-2009) for research and publications in signal processing.
\end{biographynophoto}

\begin{biographynophoto}{Mich\`ele Wigger}
(S'05, M'09) received the M.Sc.\ degree in electrical engineering (with distinction) and the Ph.D.\ degree in electrical engineering both from ETH Zurich in 2003 and 2008, respectively.
In 2009, she was a Postdoctoral Researcher at the ITA Center, University of California, San Diego. Since December 2009, she has been an Assistant Professor at Telecom ParisTech, Paris, France. Her research interests are in information and communications theory; in particular in wireless networks, feedback channels, and channels with states.
\end{biographynophoto}
\end{document}